\documentclass[11pt]{amsart}
\usepackage[T1]{fontenc}
\usepackage{amssymb,amsmath, amsthm, amsfonts, tikz-cd}

\usepackage{graphicx}
\usepackage{listings}
\usepackage[margin=.75in]{geometry}
\usepackage{lstautogobble}
\usepackage{enumerate}
\usepackage[shortlabels]{enumitem}
\usepackage{thmtools}
\usepackage{thm-restate}
\usepackage{amsthm}
\usepackage{verbatim}

\usepackage{mathtools}
\usepackage{physics}
\usepackage{color}
\usepackage{hyperref}
\usepackage[capitalise]{cleveref}
\usepackage[final]{showlabels}
\usepackage[bottom]{footmisc}
\crefformat{equation}{(#2#1#3)}

\usepackage{float}
\restylefloat{table}

\usepackage{mathrsfs}
\setlist{  
  listparindent=\parindent,
  parsep=0pt,
}

\theoremstyle{plain}
\newtheorem{thm}{Theorem}[section]
\newtheorem{prop}[thm]{Proposition}
\newtheorem{lemma}[thm]{Lemma}

\theoremstyle{definition}

\newtheorem{remark}[thm]{Remark}

\numberwithin{equation}{section} %Equation numbering

\DeclarePairedDelimiter{\paren}{\lparen}{\rparen}

\DeclarePairedDelimiter{\jp}{\langle}{\rangle}

\DeclareMathOperator{\supp}{supp}

\newcommand{\M}{{\mathcal{M}}}

\newcommand{\p}{{\partial}}

\renewcommand{\d}{\delta}

\newcommand{\R}{{\mathbb{R}}}

\newcommand{\N}{{\mathbb{N}}}

\newcommand{\Z}{{\mathbb{Z}}}

\renewcommand{\P}{{\mathcal{P}}}
\newcommand{\T}{{\mathbb{T}}}

\newcommand{\G}{{\mathfrak{G}}}

\renewcommand{\L}{{\mathcal{L}}}

\newcommand{\Sc}{{\mathcal{S}}}
\newcommand{\A}{{\mathcal{A}}}
\renewcommand{\M}{{\mathcal{M}}}

\newcommand{\Qc}{\mathcal{Q}}
\newcommand{\Cc}{\mathcal{C}}
\newcommand{\Jr}{\mathfrak{J}}

\newcommand{\tl}{\tilde}

\newcommand{\D}{\Delta}

\newcommand{\ph}{\phantom{=}}
\newcommand{\nn}{\nonumber}

\newcommand{\Lr}{\mathfrak{L}}
\newcommand{\Cr}{\mathfrak{C}}

\newcommand{\ol}{\overline}
\newcommand{\ul}{\underline}

\newcommand{\ux}{\underline{x}}

\newcommand{\umu}{\underline{\mu}}
\newcommand{\ep}{\epsilon}
\newcommand{\vep}{\varepsilon}
\newcommand{\al}{\alpha}

\newcommand{\om}{\omega}

\newcommand{\E}{\mathbb{E}}

\DeclareMathOperator{\Du}{Duh}
\renewcommand{\A}{\mathcal{A}}

\newcommand{\wh}{\widehat}

\newcommand{\uxi}{\underline{\xi}}

\let\oldtocsection=\tocsection
 
\let\oldtocsubsection=\tocsubsection
 
\let\oldtocsubsubsection=\tocsubsubsection
 
\renewcommand{\tocsection}[2]{\hspace{0em}\oldtocsection{#1}{#2}}
\renewcommand{\tocsubsection}[2]{\hspace{1em}\oldtocsubsection{#1}{#2}}
\renewcommand{\tocsubsubsection}[2]{\hspace{2em}\oldtocsubsubsection{#1}{#2}}

\begin{document}

\title[Uniqueness of Solutions to the Spectral Hierarchy]{Uniqueness of Solutions to the Spectral Hierarchy in Kinetic Wave Turbulence Theory}

\author[M. Rosenzweig]{Matthew Rosenzweig}
\address{  
Department of Mathematics\\ 
Massachusetts Institute of Technology\\
Headquarters Office\\
Simons Building (Building 2), Room 106\\
77 Massachusetts Ave\\
Cambridge, MA 02139-4307}
\email{mrosenzw@mit.edu}
\thanks{M.R. is funded in part by the Simons Foundation through the Simons Collaboration on Wave Turbulence}

\author[G. Staffilani]{Gigliola Staffilani}
\address{  
Department of Mathematics\\ 
Massachusetts Institute of Technology\\
Headquarters Office\\
Simons Building (Building 2), Room 106\\
77 Massachusetts Ave\\
Cambridge, MA 02139-4307}
\email{gigliola@math.mit.edu}
\thanks{G.S. is  funded in part by  DMS-1764403 and by the Simons Foundation through the Simons Collaboration on Wave Turbulence}

\begin{abstract}
In \cite{ES2012} and \cite{CDR2018}, Eyink and Shi and Chibbaro et al., respectively, formally derived an infinite, coupled hierarchy of equations for the spectral correlation functions of a system of weakly interacting nonlinear dispersive waves with random phases in the standard kinetic limit. Analogously to the relationship between the Boltzmann hierarchy and Boltzmann equation, this \emph{spectral hierarchy} admits a special class of factorized solutions, where each factor is a solution to the wave kinetic equation (WKE). A question left open by these works and highly relevant for the mathematical derivation of the WKE is whether solutions of the spectral hierarchy are unique, in particular whether factorized initial data necessarily lead to factorized solutions. In this article, we affirmatively answer this question in the case of 4-wave interactions by showing, for the first time, that this spectral hierarchy is well-posed in an appropriate function space. Our proof draws on work of Chen and Pavlovi\'{c} \cite{CP2010} for the Gross-Pitaevskii hierarchy in quantum many-body theory and of Germain et al. \cite{GIT2020} for the well-posedness of the WKE.
\end{abstract}

\maketitle

\section{Introduction}
\subsection{Background}
The theory of \emph{wave turbulence} describes the statistical properties of ensembles of dispersives waves, such as those on the surface of the ocean \cite{Hasselmann1962I, Hasselmann1963II, Hasselmann1963III, Zakharov1999, FLP2007} or in quantum fluids \cite{KMN2014}, with weakly nonlinear interactions using the paradigm of Boltzmann's kinetic theory for dilute gasses. Starting from a microscopic description of waves of different frequencies interacting nonlinearly, the goal is to obtain an effective macroscopic description of the evolution of the energy spectrum. The development of wave turbulence theory began with work of Peierls \cite{Peierls1929} for phonons in the first half of the twentieth century. Several decades later, Hasselmnan \cite{Hasselmann1962I, Hasselmann1963II, Hasselmann1963III} independently proposed a different theory for water waves. Zakharov and collaborators \cite{ZLF2012} breathed new life into the subject with the discovery of the \emph{Kolmogorov-Zakharov spectra}, which are stationary solutions of power-law type for such kinetic equations and which are analagous to the Kolmogorov spectra in the theory of hydrodynamic turbulence (e.g. see \cite{Frisch1995}). In the years since, wave turbulence theory has developed into a large body of literature covering many physical systems of interest, a sense of which the reader can glean from the reviews of Newell and Rumpf \cite{NR2013} and Nazarenko \cite{Nazarenko2011}.

A fundamental object in wave turbulence theory is the \emph{wave kinetic equation (WKE)}, which one should think of as the wave analogue of the Boltzmann equation for systems of particles \cite{Spohn2008}. In principle, this nonlinear equation describes the evolution of the system's energy density, alternatively Fourier space mass density. However, providing a rigorous derivation of the WKE from the underlying physical system is a very difficult problem in mathematical physics that, generally speaking, remains open.

The mathematical community's attention on this derivation problem has largely focused on the justification of the WKE starting from dynamics governed by the \emph{nonlinear Schr\"odinger equation (NLS)}
\begin{equation}
\begin{cases}
i\p_t v +\frac{1}{2\pi}\D v - |v|^2v = 0,\\
v|_{t=0} = v_0,
\end{cases}
\end{equation}
which corresponds to 4-wave interactions and which is the primary interest of this article. As the sign of the nonlinearity does not play a role due to the weakness of the interactions we consider, we have chosen the defocusing equation. Let us quickly sketch, following the argument of \cite{BGHS2020} (see Section 2 of that work for more details), how one arrives at the WKE for this model.

Suppose that one considers the NLS on the standard torus $\T_L^d$ which we identify with the box $[0,L]^d$ with periodic boundary conditions. To track the strength of the nonlinear interaction, one can introduce a new unknown through $v=\lambda u$ so that the equation above becomes
\begin{equation}
\label{eq:nls}
\begin{cases}
i\p_t u + \frac{1}{2\pi}\D u - \lambda^2 |u|^2u =0,\\
u|_{t=0} = u_0.
\end{cases}
\end{equation}
We think of $u$ as bounded in some suitable norm. Take initial data $u_0$ to be of \emph{random phase (RP)} type:
\begin{equation}
\label{eq:ID}
u_0(x) = L^{-d}\sum_{k\in\Z_L^d} \wh{u_0}(k) e^{2\pi i k\cdot x}, \qquad \wh{u_0}(k) = \sqrt{\phi(k)} e^{2\pi i\vartheta_k},
\end{equation}
where the phases $\vartheta_k$ are independent and uniformly distributed in $[0,1]$, $\phi:\R^d\rightarrow [0,\infty)$ is a deterministic Schwartz function, $\Z_L^d \coloneqq L^{-1}\Z^d$ denotes the Fourier dual of $\T_L^d$, and $\wh{\cdot}$ denotes the Fourier transform. Note that $\|u_0\|_{L^2} \sim 1$ uniformly in $L$ due to our convention for the Fourier transform. Rewriting the NLS \eqref{eq:nls} in terms of a system of ODEs for the Fourier modes $\wh{u}_k$ of $u$ and removing the linear dynamics by considering the profile $a_k(t) \coloneqq u_k(t)e^{2\pi it|k|^2}$, we obtain the integral equation
\begin{equation}
\label{eq:ak}
a_k(t) = a_k^0 - \frac{i\lambda^2}{L^{2d}}\int_0^t \sum_{{k_1,k_2,k_3\in \Z_L^d}\atop {k+k_1-k_2-k_3=0}} \ol{a_{k_1}}a_{k_2}a_{k_3}e^{2\pi i s(|k|^2+|k_1|^2-|k_2|^2-|k_3|^2)}ds. 
\end{equation}
Wave turbulence theory asserts that in the large box limit $L\rightarrow\infty$ and weakly nonlinear limit $\al\coloneqq\lambda^2 L^{-d}\rightarrow 0$, the quantity
\begin{equation}
f_k^L(t) \coloneqq \E |a_k(t)|^2 ,
\end{equation}
where the expectation $\E$ is taken with respect to the ensemble of initial data, should obey the approximation
\begin{equation}
\label{eq:approx}
f_k^L(t) \approx f(\frac{t}{T_{kin}}, k), \qquad \forall k\in\Z_L^d,
\end{equation}
where $T_{kin} \sim \alpha^{-2}$ is the \emph{kinetic timescale} at which one expects the approximation \eqref{eq:approx} to be valid and $f: \R_+\times\R^d \rightarrow\R_+$ solves the wave kinetic equation
\begin{equation}
\label{eq:WKE}
\begin{split}
\p_t f(\tau,k) &= \Cc(f)(\tau,k),\\
\Cc(f)(\tau,k) &= \int_{(\R^d)^3}dk_1dk_2dk_3\d(k+k_1-k_2-k_3)\d(|k|^2+|k_1|^2-|k_2|^2-|k_3|^2) \\
&\ph \left[f(\tau,k_2)f(\tau,k_3)(f(\tau,k_1)+f(\tau,k)) - f(\tau,k)f(\tau,k_1)(f(\tau,k_2)+f(\tau,k_3))\right],
\end{split}
\end{equation}
where $\Cc$ is called the \emph{collision operator}, which is evidently nonlinear in $f$. The Dirac expression in the collision integral of \eqref{eq:WKE} should be understood in the sense of measure. $\d(k+k_1-k_2-k_3)$ corresponds to the convolution integration over $k_2+k_3-k_1=k$, while $\d(|k|^2+|k_1|^2-|k_2|^2-|k_3|^2)$ is absolutely continuous with respect to the surface measure on the zero set
\begin{equation}
|k|^2+|k_1|^2-|k_2|^2-|k_3|^2=0.
\end{equation}

To our knowledge, the mathematical literature on the WKE derivation begins with the work of Lukkarinen and Spohn \cite{LS2011} (see also \cite{LS2009}), who showed that the evolution of the correlations for the discrete NLS at equilibrium are governed by a linearized wave kinetic equation. The first work to tackle the derivation away from equilibrium is by Buckmaster, Germain, Hani, and Shatah \cite{BGHS2021} (see also \cite{BGHS2020}). More recently, Collot and Germain \cite{CG2019, CG2020} and Deng and Hani \cite{DH2019} have independently improved the timescale over which the approximation \eqref{eq:approx} holds. In caricature, these results all show an approximation of the form
\begin{equation}
\label{eq:approxBGHS}
f_k^L(t) = \phi(k) + \frac{t}{T_{kin}}\Cc(\phi)(k) + o_{\ell_k^\infty}(\frac{t}{T_{kin}})_{L\rightarrow\infty}, \qquad \forall k\in\Z_L^d, \enspace L^{\delta} \leq t\leq T,
\end{equation}
where $\delta>0$, $T\ll T_{kin}$, and $\lambda$  scales with $L$ so that $\alpha\rightarrow 0$ as $L\rightarrow\infty$. Note that because the time $T$ is much small than $T_{kin}$, in particular $T/T_{kin}\rightarrow 0$ as $L\rightarrow\infty$, the WKE \eqref{eq:WKE} effects only a negligible change to the initial data. Thus, the right-hand side of \eqref{eq:approxBGHS} is equivalent to replacing $\phi$ with $f$. The above approximation \eqref{eq:approxBGHS} fails to reach the conjectured optimal timescale $T_{kin}$, and to the best of our knowledge, this remains a challenging open problem.

We also mention that work has been done for 3-wave and 6-wave interactions. Faou \cite{Faou2020} showed the derivation of the linearized 3-wave kinetic equation, analogous to the aforementioned work of Lukkarinen and Spohn \cite{LS2011}. Very recently, the second co-author together with Tran \cite{ST2021} have, for the first time, shown the validity of the wave kinetic equation for a system with 3-wave interactions on the conjectured optimal time scales. Lastly, de Suzzoni \cite{deS2020} has studied the correlations of the Fourier coefficients for the quintic NLS in the large box and weakly nonlinear limits.

Much more is known, mathematically speaking, about solutions to the WKE \eqref{eq:WKE}. Escobedo and Vel\'{a}zquez \cite{EV2015} first proved the local existence and uniqueness of classical solutions when the solution is radial, in which case equation \eqref{eq:WKE} reduces to a one-dimensional Boltzmann equation. They also prove the existence of global measured-valued weak solutions, a framework which allows for a ``condensation'' phenomenon where a point mass develops at the origin, and study the transfer of energy to high frequencies for large times (see also \cite{KV2015, KV2016}). More recently, Germain, Ionescu, and Tran \cite{GIT2020} have shown local well-posedness of classical solutions to \eqref{eq:WKE} without the radial assumption and in a class of function spaces which is essentially optimal with respect to the scaling of the equation \eqref{eq:WKE}. The result of Germain et al. in fact holds for a large class of dispersion relations $\omega(k)$, including the Schr\"odinger relation $\om(k)=|k|^2$, in contrast to that of \cite{EV2015}.

\subsection{The spectral hierarchy}
Several years ago, Chibbaro, Dematteis, and Rondoni \cite{CDR2018, CDJR2017} proposed an alternative, but still formal, approach to the derivation of the WKE from 4-wave Hamiltonian systems, following earlier work by Eyink and Shi \cite{ES2012} for systems with 3-wave interactions, which draws a close parallel to Boltzmann's kinetic theory. In particular, their work demonstrates the analogue of Boltzmann's \emph{propagation of molecular chaos} in the derivation of the wave kinetic equation. It also suggests that an approach to rigorously deriving the WKE based on hierarchies of equations for statistical observables, widely used in kinetic theory \cite{Lanford1975, Lanford1976, King1975, Spohn1981, IP1986, GSrT2013, PSS2014, AP2019} and quantum many-body systems \cite{ABGT2004, AGT2007, ESY2006, ESY2007, ESY2009, ESY2010, KSS2011, CP2011, CP2014, CT2014, GSS2014, Sohinger2015, CH16, CH2016-2, CH2017, CH2019, MNPRS2020}, might also be used for wave turbulence.

Specializing to our setting of \eqref{eq:nls}, Chibbaro et al. consider the \emph{empirical spectrum} defined by
\begin{equation}
\label{eq:esdef}
f_{L,\lambda}(t,\xi) \coloneqq L^{-d}\sum_{k\in \Z_L^d} |a_k(t)|^2\d_{k}(\xi), \qquad (t,\xi)\in\R_+\times\R^d,
\end{equation}
where $\d_k$ denotes the Dirac mass centered at $k\in\Z_L^d$ and the $a_k(t)$ are as in \eqref{eq:ak}. Note that $f_{L,\lambda}$ is a random Borel measure on $\R^d$, and that by conservation of mass for the NLS,
\begin{equation}
\int_{\R^d}f_{L,\lambda}(t,\xi)d\xi= L^{-d}\sum_{k\in\Z_L^d} \phi(k).
\end{equation}
The empirical spectrum is analogous to Klimontovich's empirical measure in kinetic theory \cite{Klimontovich2013}. Equipped with the empirical spectrum, one can define \emph{spectral correlation functions}
\begin{equation}
\label{eq:cordef}
f_{L,\lambda}^{(m)}(t,\xi_1,\ldots,\xi_m) \coloneqq \E\left[f_{L,\lambda}(t)^{\otimes m}(\xi_1,\ldots,\xi_m)\right], \qquad m\in\N,
\end{equation}
which satisfy a complicated coupled infinite system of equations. Note that $f_{L,\lambda}^{(m)}$ is symmetric under permutation of mode labels. After a difficult formal calculation involving a perturbative expansion of the Fourier coefficients $a_k$ and phase averaging using Feynman-Wyld diagrams, Chibbaro et al. find that in the large box limit $L\rightarrow \infty$ and weak nonlinearity limit $\alpha=\lambda^2 L^{-d}\rightarrow 0$, the (appropriately scaled) functions $f_{L}^{(m)}$ converge to a solution $f^{(m)}$ of the \emph{spectral hierarchy}: 
\begin{equation}
\label{eq:WKE_hier}
\begin{split}
\p_t f^{(m)}(t,\xi_1,\ldots,\xi_m) &= \sum_{j=1}^m \int_{(\R^3)^3}d\xi_2'd\xi_3'd\xi_4'\d(\xi_1+\xi_2'-\xi_3'-\xi_4')\d(|\xi_1|^2+|\xi_2'|^2-|\xi_3'|^2-|\xi_4'|^2) \\
&\ph\Big(f^{(m+2)}(t,\xi_1,\ldots,\xi_{j-1},\xi_{j+1},\ldots,\xi_m, \xi_2',\xi_3',\xi_4') +f^{(m+2)}(t,\xi_1,\ldots,\xi_m, \xi_3',\xi_4') \\
&\ph\ph - f^{(m+2)}(t,\xi_1,\ldots,\xi_m, \xi_2',\xi_4') - f^{(m+2)}(t,\xi_1,\ldots,\xi_m, \xi_2',\xi_3')\Big).
\end{split}
\end{equation}
That the evolution of $f^{(m)}$ is coupled to that of $f^{(m+2)}$ reflects the fact that we are dealing with 4-wave interactions.

One should note the strong similarity between the spectral hierarchy and the Boltzmann hierarchy derived by Lanford \cite{Lanford1975, Lanford1976} from the BBGKY hierarchy in the derivation of the Boltzmann equation for hard spheres in the Boltzmann-grad limit. Just as solutions to the Boltzmann equation yield a special class of \emph{factorized} solutions to the Boltzmann hierarchy, solutions $f$ to the WKE \eqref{eq:WKE} yield a special class of factorized solutions to the spectral hierarchy by taking
\begin{equation}
\label{eq:fac_sol}
f^{(m)} \coloneqq f^{\otimes m}, \qquad \forall m\in\N.
\end{equation}
Of course, it is not at all evident that starting from factorized initial data $f_0^{\otimes m}$, the solution \eqref{eq:fac_sol} is the only possible one.

\subsection{Informal account of main results}
\label{ssec:intromr}
To the best of our knowledge, there are no results on the rigorous analysis of solutions to the spectral hierarchy \eqref{eq:WKE_hier}. We saw in the last subsection that a solution trivially exists for factorized initial data, just by taking tensor products of solutions to the WKE \eqref{eq:WKE}. However, it is an open mathematical problem, first recognized by Eyink and Shi \cite[Section 3.1]{ES2012} in the context of 3-wave interactions, whether this ``trivial'' solution is the only one. In other words, it is unknown whether solutions to the spectral hierarchy are unique. 

We emphasize that this problem of uniqueness for infinite BBGKY-type hierarchies is a highly nontrivial one, and much research has been done on this subject both for classical \cite{Lanford1975, Lanford1976, King1975, Spohn1981, IP1986, GSrT2013, PSS2014, AP2019} and quantum particle systems \cite{ESY2007, KM2008, CP2010, GSS2014, CHPS-sc, HTX2015, CHPS2015, SS2015, HTX2016, HS2016, CH2020, ALR2020}. In an interacting system, the evolution of the $m$-th component of the hierarchy is coupled to the evolution of some higher components of the hierarchy--in our case, the $m$-th is coupled to the $(m+2)$-th. Consequently, it is not enough to consider just one component in the hierarchy; one needs to consider the entire infinite vector $F=(f^{(m)})_{m=1}^\infty$ in an appropriate function space. Additionally, the reader will note from the expression for the collision integral in \eqref{eq:WKE_hier} that the number of terms in the right-hand for the equation of $f^{(m)}$ grows like $m$. Since a direct estimation of the right-hand side in terms of a bound on $f^{(m+2)}$ will pick up this combinatorial factor, one needs to find a mechanism to compensate for this growth. 

This article resolves these difficulties by showing, for the first time, the local well-posedness of the spectral hierarchy \eqref{eq:WKE_hier} in the sense of Hadamard. Below we give an informal statement of our main results. We defer the precise statement \cref{thm:main} until \cref{sec:MR}, so as not to burden the reader with notation.
\begin{thm}[Informal main theorem]
\label{thm:inf}
There exists a double-indexed nested scale of function spaces $\Lr_{s_1,\ep_1}^\infty \subset \Lr_{s_2,\ep_2}^\infty$, for  $s_1\geq s_2$ and $\ep_1\geq \ep_2$, such that if the initial datum $F_0=(f_0^{(m)})_{m=1}^\infty \in \Lr_{s,\ep_1}^\infty$, then there exists a time $T$ and index $\ep_2\leq \ep_1$, both depending on the data $(s,\ep_1)$ for which there exists a unique solution to the spectral hierarchy \eqref{eq:WKE_hier} in the class $C([0,T];\Lr_{s,\ep_2}^\infty)$. Moreover, if the initial data $F_0$ and $G_0$ are close in $\Lr_{s,\ep_1}^\infty$, then their respective solutions remain close in $\Lr_{s,\ep_2}^\infty$ for short times.
\end{thm}

\begin{remark}
We have presented our results only for dimension $d=3$; however, they are valid in all dimensions $d\geq 2$ \emph{mutatis mutandis}.
\end{remark}

\begin{remark}
Implicit in the statement of \cref{thm:inf} is that the collision integral in the right-hand side of \eqref{eq:WKE_hier} is well-defined for data in the space $\Lr_{s,\ep}^\infty$. See \cref{ssec:colhier} for details.
\end{remark}

Since we show uniqueness, we also show that if a solution starts factorized, then it remains factorized for positive time. Our result implies that if we have the initial convergence in the sense of measures,
\begin{equation}
\forall m\in\N, \qquad f_{L,\lambda}^{(m)}(0) \xrightharpoonup[L\rightarrow\infty]{*} f_0^{\otimes m}
\end{equation}
and we also know that $(f_{L,\lambda}^{(m)})_{m=1}^\infty$ converges in the large box and weakly nonlinear limits for positive times to a solution $(f^{(m)})_{m=1}^\infty$ of the spectral hierarchy \eqref{eq:WKE_hier} in the function space we consider, then
\begin{equation}
\forall m\in\N, \qquad f_{L,\lambda}^{(m)}(T_{kin}t) \xrightharpoonup[L\rightarrow\infty, \alpha\rightarrow 0]{*} f(t)^{\otimes m} \qquad \forall t>0,
\end{equation}
where $f$ is the unique solution to the WKE \eqref{eq:WKE} with initial datum $f_0$. Of course, it is not obvious and, in fact, currently unknown whether the spectral correlation functions converge in the kinetic limit to solutions in the class $C([0,T]; \Lr_{s,\ep}^\infty)$. Moreover, the preceding reasoning takes for granted that solutions to the NLS \eqref{eq:nls} exist on the relevant kinetic time scale for the above statements to make sense; see \cite{BGHS2021, DH2019, CG2019, CG2020} for more discussion of this difficulty. Resolving these questions, though, is not the subject of this article.

It is an interesting question whether there is ever a need to consider solutions to the spectral hierarchy itself, as opposed to the wave kinetic equation, outside of this derivation problem. Spohn \cite[Section 5]{Spohn1981} previously raised a similar question for the relationship between the Boltzmann equation and Boltzmann hierarchy, which to the best of our knowledge remains unanswered. In the context of wave turbulence--analogous to the situation for the kinetic theory of dilute gasses--Eyink and Shi \cite[Sections 3.1.2, 4.2]{ES2012} argue that the most general ``statistically realizable'' solutions to the spectral hierarchy \eqref{eq:WKE_hier} take the form
\begin{equation}
\label{eq:SHss}
f^{(m)} = \int_{\M_+(\R^3)}d\rho(f_0) f^{\otimes m},
\end{equation}
where $\rho$ is some probability measure on the space $\M_+(\R^3)$ of positive Borel measures on $\R^3$ and $f$ denotes the solution to the WKE \eqref{eq:WKE} with initial datum $f_0$. Solutions of the form \eqref{eq:SHss} are statistical superpositions of factorized solutions to the spectral hierarchy, which Eyink and Shi call \emph{super-statistical solutions}. An equivalent interpretation is that solutions of the form \eqref{eq:SHss} correspond to solutions of the WKE with \emph{random} initial data, where the randomness is expressed through the measure $\rho$. Eyink and Shi further argue that super-statistical solutions offer a possible explanation for the phenomena of intermittency and non-Gaussian statistics in wave turbulence. The physical implications of super-statistical solutions are beyond the scope of this article. But in \cref{sec:SS}, we conclude with some comments on the implications of our work for Eyink and Shi's assertions. 

\subsection{Comments on the proof}
\label{ssec:intropf}
Let us now make a few remarks about the proof of \cref{thm:inf}. At a high level, our proof has four basic ingredients, itemized below:
\begin{enumerate}[(i)]
\item\label{item:Duh}
Iterated Duhamel expansion;
\item\label{item:fs}
A good choice of function spaces;
\item\label{item:res}
Estimates for resonant manifold integrals;
\item\label{item:com}
Combinatorial analysis.
\end{enumerate}
Starting with \ref{item:Duh}, let us consider the entire infinite vector $F=(f^{(m)})_{m=1}^\infty$, and not just its component $f^{(m)}$. Introducing the notation $\uxi_{i;j} \coloneqq (\xi_i,\ldots,\xi_j)$, for $j\geq i$, we can define the \emph{collision operator} $\Cr$ through
\begin{equation}
\label{eq:Chidef}
\begin{split}
(\Cr F)^{(m)}(\uxi_{1;m}) &\coloneqq \sum_{j=1}^m\int_{(\R^3)^3}d\uxi_{2;4}'\d(\xi_1+\xi_2'-\xi_3'-\xi_4')\d(|\xi_1|^2+|\xi_2'|^2-|\xi_3'|^2-|\xi_4'|^2) \\
&\ph\Bigg(f^{(m+2)}(\uxi_{1;j-1},\uxi_{j+1;m},\xi_2',\xi_3',\xi_4') + f^{(m+2)}(\uxi_{1;m},\xi_3',\xi_4') \\
&\ph\ph - f^{(m+2)}(\uxi_{1;m},\xi_2',\xi_4') - f^{(m+2)}(\uxi_{1;m},\xi_2',\xi_3') \Bigg).
\end{split}
\end{equation}
The idea now is to rewrite the WKE in integral, alternatively mild, form
\begin{equation}
\label{eq:introDuh}
F(t) = F_0 + \int_0^t \Cr[F(\tau)]d\tau.
\end{equation}
The integration with respect to time commutes with the collision operator $\Cr$ by Fubini-Tonelli. Thus, we can iterate the equation \eqref{eq:introDuh} $j$ times, for arbitrary $j\in\N_0$, to obtain that if $F$ satisfies \eqref{eq:introDuh}, then it must also satisfy the equation
\begin{equation}
\label{eq:introDuhj}
F(t)=\sum_{k=0}^{j} \frac{t^k}{k!}\Cr^k[F_0] + \int_0^t\int_0^{t_1}\cdots\int_0^{t_j}\Cr^{j+1}[F(t_{j+1})]dt_{j+1}\cdots dt_2dt_1.
\end{equation}
Proceeding formally, we can iterate the Duhamel expansion infinitely many times to obtain that the unique solution $F$ to equation \eqref{eq:introDuh} \emph{should} be given by the \emph{Duhamel series}
\begin{equation}
\label{eq:Duhs}
\sum_{j=0}^\infty \frac{t^j}{j!}\Cr^j[F_0].
\end{equation}
Of course, it is not at all obvious that there is a suitable function space in which the series \eqref{eq:Duhs} converges, and our passage from \eqref{eq:introDuhj} to \eqref{eq:Duhs} assumed that the error term in \eqref{eq:introDuhj}
$$ \int_0^t\int_0^{t_1}\cdots\int_0^{t_j}\Cr^{j+1}[F(t_{j+1})]dt_{j+1}\cdots dt_2dt_1$$
vanishes as $j\rightarrow\infty$.

This leads us to ingredient \ref{item:fs}. Inspired by the choice of norms introduced by Chen and Pavlovi\'{c} \cite[Equation (10)]{CP2010} to study the Cauchy problem for the Gross-Pitaevskii hierarchy, the infinite-particle limit of the quantum BBGKY hierarchy widely used in the derivation of NLS-type equations as effective descriptions of interacting Bose gases \cite{ABGT2004, AGT2007, ESY2006, ESY2007, ESY2009, ESY2010, KSS2011, CP2011, CP2014, CT2014, GSS2014, Sohinger2015, CH16, CH2016-2, CH2017, CH2019, MNPRS2020}, we use a norm of the form (see \eqref{eq:hnorm} for the precise definition)
\begin{equation}
\label{eq:intronorm}
\|F\|_{\Lr_{s,\ep}^\infty} \coloneqq \sum_{m=1}^\infty \ep^m \|f^{(m)}\|_{\L_{s,m}^\infty}, \qquad  F=(f^{(m)})_{m=1}^\infty.
\end{equation}
Here, $\L_{s,m}^\infty$ is a weighted $L^\infty$ space, defined in \eqref{eq:knorm}, where the parameter $s$ dictates how fast the function $f^{(m)}$ decays at infinity. The parameter $\ep$ is introduced so as to make the infinite sum in \eqref{eq:intronorm} finite. A good example to keep in mind is that if $f^{(m)}=f^{\otimes m}$ for all $m$, then the above norm is finite for any choice $\ep<\|(1+|\xi|^2)^{s/2}f\|_{L^\infty}^{-1}$. Lastly, the reader should note that the definition of $\L_{s,m}^\infty$ is quite different from that of Chen and Pavlovi\'{c}, who use Sobolev norms. Rather, the inspiration we take comes from the role of the small parameter $\ep$.

With our scale of function spaces $\Lr_{s,\ep}^\infty$, we first want to prove that the series \eqref{eq:Duhs} converges to a solution of the spectral hierarchy whose norm quantitatively depends on the initial datum. Then we want to show that this convergent series is the only solution, i.e. it is unique. This leads us to ingredients \ref{item:res} and \ref{item:com}. To show the convergence (see \cref{ssec:WPcon}), we seek to estimate each of the summands in \eqref{eq:Duhs} in the space $\Lr_{s,\ep}^\infty$ introduced above. Each term in the collision integral in \eqref{eq:WKE_hier} is an integration of $f^{(m+2)}$ over a manifold, the so-called resonant manifold, which produces a function of $m$ variables that we need to estimate in $\L_{s,m}^\infty$. Such integrals were previously analyzed by Germain et al. \cite[Section 3]{GIT2020} for the nonlinear collision integral \eqref{eq:WKE} to prove well-posedness of the WKE.\footnote{Related results were also shown by Lukkarinen \cite{Lukkarinen2007} for the asymptotics of resolvent integrals for lattice dispersion relations.} By revisiting the work of Germain et al. and exploiting the favorable structure of the norm $\|\cdot\|_{\Lr_{s,\ep}^\infty}$ (see \cref{sec:col} for details), we are able to estimate such expressions at the level of the spectral hierarchy collision integral. Summing up all of our collision integral estimates entails a delicate balance (see \eqref{eq:bal}) between the parameters. It is essential that we have the freedom to choose $\ep_2, T$ sufficiently small and that we use the $1/k!$ decay in \eqref{eq:Duhs} to obtain a convergent expression. The uniqueness step (see \cref{ssec:WPu}) proceeds similarly, using our already established collision estimates and combinatorial analysis. The continuous dependence on the initial data (see \cref{ssec:WPdep}) follows from the bound established for the solution in the convergence step and from the linearity of the spectral hierarchy.

\begin{remark}
Germain et al. \cite{GIT2020} also establish local well-posedness of the WKE \eqref{eq:WKE} in weighted $L^2$ spaces. The arguments of our article do not suffice to establish an analogous result at the level of the spectral hierarchy, due to the delicate cancellation in the collision integral that one would need to exploit. It is an interesting mathematical question to address this difficulty, which we hope to do in future work.
\end{remark}

\subsection{Organization of article}
We briefly comment on the organization of the remainder of the article. In \cref{sec:MR}, we introduce notation specific to our paper (e.g. for function spaces) and give the precise statement of our main theorem. In \cref{sec:col}, we study the collision operators appearing in the right-hand sides of equations \eqref{eq:WKE} and \eqref{eq:WKE_hier}. We begin in \cref{ssec:colWKE} with a review of the boundedness of the nonlinear collision operator in weighted $L^\infty$ spaces as discussed in \cite{GIT2020}, focusing on the role of estimates for integrals over the resonant manifold. We then use this analysis in \cref{ssec:colhier} together with tensorization arguments to prove analogous bounds in weighted $L^\infty$ spaces suitable for hierarchies. In \cref{sec:WP}, we give the proof of our main result, \cref{thm:main}. This section is divided into three subsections corresponding to existence, uniqueness, and continuous dependence on the initial data. Lastly, \cref{sec:SS} contains some remarks on the implications of our \cref{thm:main} for super-statistical solutions.

\subsection{Notation}
In this last subsection of the introduction, we introduce the notation used in the body of the article without further comment.

Given two quantities $A,B\geq 0$, we write $A\lesssim B$ if there exists a constant $C>0$ such that $A\leq CB$. If $A \lesssim B$ and $B\lesssim A$, we write $A\sim B$. To emphasize the dependence of the constant $C$ on some parameter $p$, we sometimes write $A\lesssim_p B$ or $A\sim_p B$.

We denote the natural numbers excluding zero by $\N$ and including zero by $\N_0$. Similarly, we denote the nonnegative real numbers by $\R_{\geq 0}$ and the positive real numbers by $\R_+$. Given $N\in\N$ and points $x_{1},\ldots,x_{N}$ in some set $X$, we will write $\ux_N$ to denote the $N$-tuple $(x_{1},\ldots,x_{N})$. We use the notation $\jp{x}\coloneqq (1+|x|^2)^{1/2}$ to denote the Japanese bracket. Given a function $f$ and positive integer $m$, we let $f^{\otimes m}$ denote the $m$-fold tensor product of $f$.

We denote the space of nonnegative Borel measures on $\R^n$ by $\M_+(\R^n)$. We denote the subspace of probability measures (i.e. elements $\mu\in\M_+(\R^n)$ with $\mu(\R^n)=1$) by $\P(\R^n)$. When $\mu$ is in fact absolutely continuous with respect to Lebesgue measure on $\R^n$, we shall abuse notation by writing $\mu$ for both the measure and its density function. We denote the Banach space of complex-valued continuous, bounded functions on $\R^n$ by $C(\R^n)$ equipped with the uniform norm $\|\cdot\|_{\infty}$. More generally, we denote the Banach space of $k$-times continuously differentiable functions with bounded derivatives up to order $k$ by $C^k(\R^n)$ equipped with the natural norm, and we define $C^\infty \coloneqq \bigcap_{k=1}^\infty C^k$. We denote the Schwartz space of functions by $\Sc(\R^n)$ and the space of tempered distributions by $\Sc'(\R^n)$. For $p\in [1,\infty]$ and $D\subset\R^n$, we define $L^p(D)$ to be the usual Banach space equipped with the norm
\begin{equation}
\|f\|_{L^p(D)} \coloneqq \paren*{\int_D |f(x)|^p dx}^{1/p}
\end{equation}
with the obvious modification if $p=\infty$. When the underlying domain is clear from context, we often will just write $\|f\|_{L^p}$.

\section{Main results}
\label{sec:MR}
In the introduction, we have only considered waves with the classic Schr\"odinger dispersion relation $\om(k)=|k|^2$; however, our \cref{thm:main} in fact holds for the larger class of spherically symmetric relations considered in \cite{GIT2020}. The precise assumptions on $\om(k)=\Omega(|k|)$ are stated below.
\begin{enumerate}[(i)]
\item
\label{item:asmpreg}
$\Omega \in C_{loc}^1(\R_+), \Omega \geq 0$;
\item
\label{item:asmplb}
$\Omega'(x)\geq c_1 x$ for all $x\in\R_+$, for some $c_1>0$;
\item
\label{item:asmpdub}
$\Omega(x)\leq \Omega(c_2 x)/2$ for all $x\in \R_+$, for some $c_2>0$.
\end{enumerate}
Besides the Schr\"odinger relation, examples in this class include the Bogoliubov dispersion law \cite{Eckern1984}
\begin{equation}
\om(k) = \sqrt{\theta_1|k|^2+\theta_2|k|^4}, \qquad \theta_1,\theta_2>0,
\end{equation}
the modified Bogoliubov/Bohm-Pines dispersion law \cite{BP1951}
\begin{equation}
\om(k) = \sqrt{\theta_0 + \theta_1|k|^2+\theta_2|k|^4}, \qquad \theta_0,\theta_1,\theta_2>0,
\end{equation}
and their low-temperature approximations \cite{EPV2011, ItG2001}
\begin{equation}
\om(k) = \lambda_0 + \lambda_1|k|^2 + \lambda_2|k|^4, \qquad \lambda_0(\theta_0), \lambda_1(\theta_1), \lambda_2(\theta_2).
\end{equation}

Next, in order to give the rigorous version of \cref{thm:inf}, we need to introduce some notation for the scale of function spaces in which our well-posedness result holds. For $m\in\N$ and $s\geq 0$, we define
\begin{equation}
\label{eq:knorm}
\|f^{(m)}\|_{\L_{s,m}^\infty} \coloneqq \|\jp{\xi_1}^s\cdots \jp{\xi_m}^{s} f^{(m)}\|_{L^\infty(\R^{3m})},
\end{equation}
and we define the space
\begin{equation}
\L_{s,m}^\infty \coloneqq \{f^{(m)} \in C^0(\R^{3m}) : \|f^{(m)}\|_{\L_{s,m}^\infty} < \infty\},
\end{equation}
which the reader may check is Banach. Using the above $m$-mode norms, we build a norm for hierarchies $F = (f^{(m)})_{m=1}^\infty$ as follows. For $\ep\in (0,1)$, we define
\begin{equation}
\label{eq:hnorm}
\|F\|_{\Lr_{s,\ep}^\infty} \coloneqq \sum_{m=1}^\infty \ep^m \|f^{(m)}\|_{\L_{s,m}^\infty}
\end{equation}
and $\Lr_{s,\ep}^\infty$ as the subset of $\prod_{m=1}^\infty \L_{s,m}^\infty$ for which the right-hand side is finite. It is straightforward to check that $(\Lr_{s,\ep}^\infty, \|\cdot\|_{\Lr_{s,\ep}^\infty})$ is a Banach space.

Finally, we are prepared to give the precise version of \cref{thm:inf} above.
\begin{thm}
\label{thm:main}
Consider initial data $F_0=(f_0^{(m)})_{m=1}^\infty \in \Lr_{s,\vep_1}^\infty$, for $0<\ep_1<1$ and $s>2$, such that each $f^{(m)}$ is symmetric under permutation of coordinate labels. Then there exists a parameter $0<\ep_2(s,\ep_1) < \ep_1$ and time $T(s,\ep_1)>0$, such that the unique solution $F\in C([0,T];\Lr_{s,\ep_2}^\infty)$ to the equation \eqref{eq:introDuh} with initial datum $F_0$ is given by the convergent series
\begin{equation}
F(t) = \sum_{j=0}^\infty \frac{t^j}{j!}\Cr^j[F_0].
\end{equation}
Moreover, we have the bound
\begin{equation}
\label{eq:solbnd}
\sup_{0\leq t\leq T} \|F(t)\|_{\Lr_{s,\ep_2}^\infty} \lesssim_{s,\ep_1,\ep_2,T} \|F_0\|_{\Lr_{s,\ep_1}^\infty}.
\end{equation}
\end{thm}

\begin{remark}
If the initial datum $F_0$ has nonnegative components, the explicit form of the Duhamel series shows that on $[0,T]$, the solution also has nonnegative components.
\end{remark}

\begin{remark}
Since the equation \eqref{eq:WKE_hier} is linear, the bound \eqref{eq:solbnd} implies that the solution map depends Lipschitz continuously on the initial data. See \cref{ssec:WPdep}.
\end{remark}

\begin{remark}
By the fundamental theorem of calculus, if $F \in C([0,T];\Lr_{s,\ep_1}^\infty)$ solves equation \eqref{eq:introDuh}, then $F\in C^k([0,T];\Lr_{s,\ep_2}^\infty)$, for any $0<\ep_2<\ep_1$ and $k\in\N$; i.e. $F$ is smooth in time. See the beginning remarks of \cref{sec:WP} for more details.
\end{remark}

\section{The collision operator}
\label{sec:col}
In this section, we review the boundedness of the collision operator for the WKE \eqref{eq:WKE} and spectral hierarchy \eqref{eq:WKE_hier}. We will then use the properties established in this section to study the Cauchy problem for these equations. Our strategy is to first understand the \emph{nonlinear} collision operator appearing in the right-hand side of \eqref{eq:WKE} as the restriction of a trilinear operator, where all three arguments take the same input. This part of the analysis has essentially been completed by Germain et al. \cite[Section 3]{GIT2020} in their proof of local well-posedness for the WKE, so we only sketch the details. With this understanding, we will then use tensorization arguments (i.e. fixing a subset of coordinates) to reduce the boundedness of the \emph{linear} collision operator in the right-hand side of \eqref{eq:WKE_hier} to that of the aforementioned multilinear forms. The main challenge is the identification of a good choice of norms to use for the hierarchy collision operator.

\subsection{WKE}
\label{ssec:colWKE}
In order to avoid repeatedly writing out the collision integral in the WKE \eqref{eq:WKE}, we introduce a more compact way of writing this equation. Let us write
\begin{equation}
\om \coloneqq \om(\xi), \quad \om_i \coloneqq \om(\xi_i), \quad f \coloneqq f(\xi), \quad f_i\coloneqq f(\xi_i).
\end{equation}
With this notation, the nonlinear collision operator $\Cc$ becomes
\begin{equation}
\label{eq:Cdef}
\begin{split}
\Cc[f] = \int_{(\R^3)^3}d\uxi_{1;3}\d(\xi+\xi_1-\xi_2-\xi_3)\d(\om + \om_1-\om_2-\om_3) \paren*{f_2f_3(f_1+f)-ff_1(f_2+f_3)}.
\end{split}
\end{equation}

As observed by Germain et al. \cite[Section 3]{GIT2020}, one can understand the collision operator $\Cc$ as the restriction of a trilinear operator through
\begin{equation}
\label{eq:Cdcomp}
\Cc[f] \coloneqq \Cc_1[f,f,f] + \Cc_2[f,f,f] -\Cc_3[f,f,f]-\Cc_4[f,f,f],
\end{equation}
where
\begin{equation}
\label{eq:C1def}
\Cc_1[f,g,h](\xi) \coloneqq \int_{(\R^3)^3}d\uxi_{3}\d(\xi+\xi_1-\xi_2-\xi_3)\d(\om + \om_1-\om_2-\om_3)f(\xi_1)g(\xi_2)h(\xi_3),
\end{equation}
\begin{equation}
\label{eq:C2def}
\Cc_2[f,g,h](\xi) \coloneqq \int_{(\R^3)^3}d\uxi_{3}\d(\xi+\xi_1-\xi_2-\xi_3)\d(\om + \om_1-\om_2-\om_3)f(\xi)g(\xi_2)h(\xi_3),
\end{equation}
\begin{equation}
\label{eq:C3def}
\Cc_3[f,g,h](\xi) \coloneqq \int_{(\R^3)^3}d\uxi_{3}\d(\xi+\xi_1-\xi_2-\xi_3)\d(\om + \om_1-\om_2-\om_3)f(\xi)g(\xi_1)h(\xi_2),
\end{equation}
\begin{equation}
\label{eq:C4def}
\Cc_4[f,g,h](\xi) \coloneqq \int_{(\R^3)^3}d\uxi_{3}\d(\xi+\xi_1-\xi_2-\xi_3)\d(\om + \om_1-\om_2-\om_3)f(\xi)g(\xi_1)h(\xi_3).
\end{equation}
The reader will recall the assumptions on the dispersion relation $\om$ from \cref{sec:MR}

For $1\leq r\leq \infty$ and $s\geq 0$, we introduce the scale of weighted $L^r$ norms
\begin{equation}
\|f\|_{L_s^r(\R^3)} \coloneqq \|\jp{x}^s f\|_{L^r},
\end{equation}
where $\jp{\cdot}$ denotes the Japanese bracket. Evidently, the completion of Schwartz functions in this norm defines a Banach space $L_s^r(\R^3)$ containing $L^r(\R^3)$. For $r=\infty$, we modify the expected definition of $L_s^\infty$ to be
\begin{equation}
L_s^\infty(\R^3) \coloneqq \{f\in C^0(\R^3) : \|f\|_{L_s^\infty} < \infty\}.
\end{equation}

Following \cite[Section 3]{GIT2020}, we sketch the proof that the collision operator $\Cc$ is well-defined by showing that $\Cc_1,\ldots,\Cc_4$ are trilinear operators bounded from $(L_s^\infty(\R^3))^3$ to $L_{s+\gamma}^\infty(\R^3)$ for $0\leq \gamma <s-2$ and $s>2$. In particular, the output of the collision operator has better decay at infinity than its inputs.

\subsubsection{Boundedness of $\Cc_1$}
We start by showing the boundedness of the operator $\Cc_1$ defined in \eqref{eq:C1def}.
\begin{prop}
\label{prop:C1}
For $s>2$ and $0\leq\gamma < \min\{s-2,1\}$, $\Cc_1$ is a well-defined bounded trilinear operator $(L_s^\infty(\R^3))^3 \rightarrow L_{s+\gamma}^\infty(\R^3)$.
\end{prop}
\begin{proof}
Let $f,g,h \in L_s^\infty(\R^3)$. Writing
\begin{equation}
\begin{split}
\jp{\xi}^{s+\gamma}\Cc_1[f,g,h](\xi) &= \int_{(\R^3)^3}d\uxi_3\d(\xi+\xi_1-\xi_2-\xi_3)\d(\om+\om_1-\om_2-\om_3)\frac{\jp{\xi}^{s+\gamma}}{\jp{\xi_1}^s\jp{\xi_2}^s\jp{\xi_3}^s} \\
&\ph\jp{\xi_1}^sf(\xi_1) \jp{\xi_2}^sg(\xi_2) \jp{\xi_3}^sh(\xi_3),
\end{split}
\end{equation}
we see that
\begin{equation}
\begin{split}
\|\Cc_1[f,g,h]\|_{L_{s+\gamma}^\infty} &\leq \|f\|_{L_s^\infty} \|g\|_{L_s^\infty} \|h\|_{L_s^\infty}\\
&\ph\times\sup_{\xi\in\R^3} \int_{(\R^3)^3}d\uxi_3\d(\xi+\xi_1-\xi_2-\xi_3)\d(\om+\om_1-\om_2-\om_3)\frac{\jp{\xi}^{s+\gamma}}{\jp{\xi_1}^s\jp{\xi_2}^s\jp{\xi_3}^s}.
\end{split}
\end{equation}
Thus, the proof of the proposition is complete, assuming the finiteness of the integral in the right-hand side, the proof of which we defer to \cref{lem:C1} below.
\end{proof}

To establish the missing estimate above, we first define the phase function
\begin{equation}
\label{eq:Gdef}
\G_{\xi,\xi_1}(z) \coloneqq \om(\xi+\xi_1-z) + \om(z) - \om(\xi) - \om(\xi_1), \qquad \forall z\in\R^3.
\end{equation}
for fixed $\xi, \xi_1\in\R^3$. We define the \emph{resonant manifold} $S_{\xi,\xi_1}$ to be the zero set of $\G_{\xi,\xi_1}$. 

\begin{lemma}
\label{lem:C1}
For $s,\gamma$ as above, there is a constant $M>0$ such that 
\begin{equation}
\begin{split}
&\sup_{\xi\in\R^3}\int_{\R^3} \int_{(\R^3)^3}d\uxi_3\d(\xi+\xi_1-\xi_2-\xi_3)\d(\om+\om_1-\om_2-\om_3)\frac{\jp{\xi}^{s+\gamma}}{\jp{\xi_1}^s\jp{\xi_2}^s\jp{\xi_3}^s} \\
&\lesssim_{s,\gamma} \sup_{\xi\in\R^3}\int_{\R^3}\int_{S_{\xi,\xi_1}} \frac{\jp{\xi}^\gamma}{\jp{\xi_1}^s\jp{z}^{s} |\nabla_z\G_{\xi,\xi_1}(z)|}d\mu(z) d\xi_1\leq M,
\end{split}
\end{equation}
where $\mu$ is the surface measure on the manifold $S_{\xi,\xi_1}$.
\end{lemma}
\begin{proof}
See Steps 2-4 in the proof of \cite[Proposition 3.1]{GIT2020}.
\end{proof}

\subsubsection{Boundedness of $\Cc_2$}
We now address the boundedness of the second operator $\Cc_2$.
\begin{prop}
\label{prop:C2}
For $s>2$ and $0\leq\gamma<s-2$, $\Cc_2$ is well-defined as a bounded trilinear operator $(L_s^\infty(\R^3))^3 \rightarrow L_{s+\gamma}^\infty(\R^3)$.
\end{prop}
\begin{proof}
Let $f,g,h\in L_s^\infty(\R^3)$. Unlike in the proof of \cref{prop:C1}, we reduce the boundnedness of $\Cc_2$ to the boundedness of the bilinear operator
\begin{equation}
\Qc[g,h](\xi) \coloneqq \int_{(\R^3)^2}d\xi_1d\xi_2 \d(\om+\om_1-\om_2-\om(\xi+\xi_1-\xi_2))g(\xi_2)h(\xi+\xi_1-\xi_2).
\end{equation}
Indeed,
\begin{align}
\jp{\xi}^{s+\gamma}\Cc_2[f,g,h](\xi) &= \jp{\xi}^{s+\gamma} f(\xi)\int_{(\R^3)^3}d\uxi_{1;3}\d(\xi+\xi_1-\xi_2-\xi_3)\d(\om+\om_1-\om_2-\om_3)g(\xi_2)h(\xi_3) \nn\\
&=\jp{\xi}^{s+\gamma} f(\xi)\Qc[g,h](\xi) \nn\\
&\leq \|f\|_{L_s^\infty} \jp{\xi}^\gamma |\Qc[g,h](\xi)|.
\end{align}
Thus, it suffices to show that
\begin{equation}
\|\Qc[g,h]\|_{L_{\gamma}^\infty} \lesssim_{\gamma,s} \|g\|_{L_s^\infty} \|h\|_{L_s^\infty}.
\end{equation}
To prove this property, we observe that
\begin{equation}
\begin{split}
\jp{\xi}^\gamma \Qc[g,h](\xi) &= \int_{(\R^3)^2}d\uxi_{1;2}\d(\om+\om_1-\om_2-\om(\xi+\xi_1-\xi_2))\frac{\jp{\xi}^\gamma}{\jp{\xi_2}^s \jp{\xi+\xi_1-\xi_2}^s} \\
&\ph\jp{\xi_2}^s g(\xi_2) \jp{\xi+\xi_1-\xi_2}^s h(\xi+\xi_1-\xi_2).
\end{split}
\end{equation}
Taking the absolute value of both sides of the preceding identity, we obtain that
\begin{equation}
\|\Qc[g,h]\|_{L_{\gamma}^\infty} \leq \|g\|_{L_s^\infty} \|h\|_{L_s^\infty} \sup_{\xi\in\R^3}\int_{(\R^3)^2}d\uxi_{1;2}\d(\om+\om_1-\om_2-\om(\xi+\xi_1-\xi_2))\frac{\jp{\xi}^\gamma}{\jp{\xi_2}^s \jp{\xi+\xi_1-\xi_2}^s}.
\end{equation}
The conclusion of the proof follows from the finiteness of the supremum in the right-hand side, which we establish with the next lemma.
\end{proof}

\begin{lemma}
\label{lem:C2}
Under the assumptions on $s$ and $\gamma$ as above, there is a constant $M>0$ such that 
\begin{equation}
\begin{split}
&\sup_{\xi\in\R^3}\int_{(\R^3)^2}d\uxi_{1;2}\d(\om+\om_1-\om_2-\om(\xi+\xi_1-\xi_2))\frac{\jp{\xi}^\gamma}{\jp{\xi_2}^s \jp{\xi+\xi_1-\xi_2}^s} \\ 
&\lesssim_{s,\gamma} \sup_{\xi\in\R^3} \int_{\R^3} \int_{S_{\xi,\xi_1}} \frac{\jp{\xi}^\gamma \jp{z}^{-s}\jp{\xi+\xi_1-z}^{-s}}{|\nabla_z\G_{\xi,\xi_1}(z)|}d\mu(z) d\xi_1 \leq M.
\end{split}
\end{equation}
\end{lemma}
\begin{proof}
See Steps 2-3 in the proof of \cite[Proposition 3.2]{GIT2020}.
\end{proof}

\subsubsection{Boundedness of $\Cc_3, \Cc_4$}
In this last sub-subsection, we address the boundedness of the operators $\Cc_3,\Cc_4$. By symmetry under swapping $\xi_2\leftrightarrow \xi_3$, it suffices to consider only $\Cc_3$.
\begin{prop}
\label{prop:C3}
For $s>2$ and $0\leq\gamma<s-2$, the operator $\Cc_3$ is well-defined as a bounded trilinear operator $L_s^\infty(\R^3)^3\rightarrow L_{s+\gamma}^\infty(\R^3)$.
\end{prop}
\begin{proof}
Similarly to as in the proof of \cref{prop:C2}, we reduce to the boundedness of the bilinear operator
\begin{equation}
\Qc[g,h] \coloneqq \int_{(\R^3)^2}d\uxi_{1;2}\d(\om+\om_1-\om_2-\om(\xi+\xi_1-\xi_2))g(\xi_1)h(\xi_2),
\end{equation}
where we have recycled the notation $\Qc$. Indeed,
\begin{align}
\jp{\xi}^{s+\gamma}\Cc_3[f,g,h](\xi) &= \jp{\xi}^{s+\gamma}f(\xi)\int_{(\R^3)^2}d\uxi_{1;2}\d(\om+\om_1-\om_2-\om(\xi+\xi_1-\xi_2)) g(\xi_1) h(\xi_2) \nn\\
&= \jp{\xi}^{s+\gamma}f(\xi)\Qc[g,h](\xi).
\end{align}
Consequently,
\begin{equation}
\|\Cc_3[f,g,h]\|_{L_{s+\gamma}^\infty} \leq \|f\|_{L_s^\infty} \|\Qc[g,h]\|_{L_\gamma^\infty}.
\end{equation}
Now observe that
\begin{align}
\jp{\xi}^{\gamma}\Qc[g,h](\xi) &= \int_{(\R^3)^2}d\uxi_{1;2}\d(\om+\om_1-\om_2-\om(\xi+\xi_1-\xi_2))\frac{\jp{\xi}^\gamma}{\jp{\xi_1}^s\jp{\xi_2}^s} \jp{\xi_1}^sg(\xi_1) \jp{\xi_2}^sh(\xi_2),
\end{align}
which implies that
\begin{equation}
\|\Qc[g,h]\|_{L_\gamma^\infty} \leq \|g\|_{L_s^\infty} \|h\|_{L_s^\infty} \sup_{\xi\in\R^3}\int_{(\R^3)^2}d\uxi_{1;2}\d(\om+\om_1-\om_2-\om(\xi+\xi_1-\xi_2))\frac{\jp{\xi}^\gamma}{\jp{\xi_1}^s\jp{\xi_2}^s}.
\end{equation}
The desired conclusion now follows from \cref{lem:C3} below.
\end{proof}

\begin{lemma}
\label{lem:C3}
For $s,\gamma$ as above, There is a constant $M>0$ such that
\begin{equation}
\begin{split}
&\sup_{\xi\in\R^3}\int_{(\R^3)^2}d\uxi_{1;2}\d(\om+\om_1-\om_2-\om(\xi+\xi_1-\xi_2))\frac{\jp{\xi}^\gamma}{\jp{\xi_1}^s\jp{\xi_2}^s} \\
&\lesssim_{s,\gamma} \sup_{\xi\in\R^3} \int_{\R^3}\int_{S_{\xi,\xi_1}} \frac{\jp{\xi}^\gamma \jp{z}^{-s}\jp{\xi_1}^{-s}}{|\nabla_z\G_{\xi,\xi_1}(z)|} d\mu(z)d\xi_1\leq M.
\end{split}
\end{equation}
\end{lemma}
\begin{proof}
See the proof of \cite[Proposition 3.3]{GIT2020}.
\end{proof}

\subsection{Spectral hierarchy}
Having reviewed the structure of the nonlinear collision operator $\Cc$ as the restriction of bounded trilinear operators, we transition to understanding the structure of the spectral hierarchy collision operator $\Cr$ as a bounded linear operator in the scale of function spaces $\Lr_{s,\ep}$ introduced in \eqref{eq:hnorm}.

\label{ssec:colhier}
\subsubsection{Decomposition}
Let us start by defining, for $m\in\N$ and integer $1\leq j\leq m$, the operator
\begin{equation}
\begin{split}
\Cr_{j;m}[f^{(m+2)}](\uxi_{1;m}) &\coloneqq \int_{(\R^3)^3}d\uxi_{2;4}'\d(\xi_1+\xi_2'-\xi_3'-\xi_4')\d(\om_{234}^1)\Bigg(f^{(m+2)}(\uxi_{1;j-1},\xi_{j+1;m},\uxi_{2;4}') \\
&\ph+ f^{(m+2)}(\uxi_{1;m}, \xi_3',\xi_4') - f^{(m+2)}(\uxi_{1;m}, \xi_2',\xi_4') - f^{(m+2)}(\uxi_{1;m}, \xi_2',\xi_3')\Bigg),
\end{split}
\end{equation}
where
\begin{equation}
\om_{234}^1 \coloneqq \om(\xi_1) + \om(\xi_2') - \om(\xi_3')-\om(\xi_4').
\end{equation}
If $F=(f^{(m)})_{m=1}^\infty$, then it follows that
\begin{equation}
\label{eq:Crdcomp}
\Cr[F]^{(m)} = \sum_{j=1}^m \Cr_{j;m}[f^{(m+2)}].
\end{equation}
Similar to the decomposition \eqref{eq:Cdcomp} for $\Cc$, we can write
\begin{equation}
\label{eq:Crjdcomp}
\Cr_{j;m} = \Cr_{1,j;m} + \Cr_{2, j;m} - \Cr_{3, j;m} - \Cr_{4, j;m},
\end{equation}
where
\begin{align}
\Cr_{1,j;m}[f^{(m+2)}](\uxi_{1;m}) &\coloneqq \int_{(\R^3)^3}d\uxi_{2;4}'\d(\xi_1+\xi_2'-\xi_3'-\xi_4')\d(\om_{234}^1) f^{(m+2)}(\uxi_{1;j-1},\uxi_{j+1;m},\uxi_{2;4}'), \label{eq:C1hi}\\
\Cr_{2,j;m}[f^{(m+2)}](\uxi_{1;m}) &\coloneqq \int_{(\R^3)^3}d\uxi_{2;4}'\d(\xi_1+\xi_2'-\xi_3'-\xi_4')\d(\om_{234}^1)f^{(m+2)}(\uxi_{1;m}, \xi_3',\xi_4'),\label{eq:C2hi}\\
\Cr_{3,j;m}[f^{(m+2)}](\uxi_{1;m}) &\coloneqq \int_{(\R^3)^3}d\uxi_{2;4}'\d(\xi_1+\xi_2'-\xi_3'-\xi_4')\d(\om_{234}^1)f^{(m+2)}(\uxi_{1;m}, \xi_2',\xi_4'), \label{eq:C3hi}\\
\Cr_{4,j;m}[f^{(m+2)}](\uxi_{1;m}) &\coloneqq \int_{(\R^3)^3}d\uxi_{2;4}'\d(\xi_1+\xi_2'-\xi_3'-\xi_4')\d(\om_{234}^1)f^{(m+2)}(\uxi_{1;m}, \xi_2',\xi_3'). \label{eq:C4hi}
\end{align}
The fact $\Cr_{2,j;m}, \Cr_{3,j;m}, \Cr_{4,j;m}$ are independent of $j$, while $\Cr_{1,j;m}$ is not, reflects that in the latter, three variables of $f^{(m+2)}$ are integrated out, while in the former, only two variables of $f^{(m+2)}$ are integrated out.

\begin{remark}
Our notation $\uxi_{1;j-1}=(\xi_1,\ldots,\xi_{j-1})$ is slightly abusive, given that if $j=1$, $\uxi_{1;0}$ is nonsensical. In this case, as with all cases where the range does not make sense, $\uxi_{1;j-1}$ should be understood as vacuous.
\end{remark}

\subsubsection{Boundedness of $\Cr_{1, j;m}$}
Let us first consider the case $m=1$. In analogy to \cref{prop:C1}, we want to show that if $f^{(3)} \in \L_{s,3}^\infty$ for suitable $s\geq 0$, then $\Cr_{1, 1;1}[f^{(3)}] \in \L_{s+\gamma,1}^\infty$, for suitable $\gamma\geq 0$, and
\begin{equation}
\label{eq:C1hi_bnd}
\|\Cr_{1, 1;1}[f^{(3)}]\|_{\L_{s+\gamma,1}^\infty} \lesssim_{s,\gamma} \|f^{(3)}\|_{\L_{s,3}^\infty}.
\end{equation}
Writing
\begin{equation}
\jp{\xi_1}^{s+\gamma}\Cr_{1, 1;1}[f^{(3)}](\xi_1) = \int_{(\R^3)^3}d\uxi_{2;4}'\d(\xi_1+\xi_2'-\xi_3'-\xi_4')\d(\om_{234}^1) \frac{\jp{\xi_1}^{s+\gamma}}{\jp{\xi_2'}^s\jp{\xi_3'}^s \jp{\xi_4'}^s} \jp{\xi_2'}^s\jp{\xi_3'}^s\jp{\xi_4'}^sf^{(3)}(\xi_{2;4}'),
\end{equation}
we see that
\begin{equation}
\|\Cr_{1, 1;1}[f^{(3)}]\|_{\L_{s+\gamma,1}^\infty} \leq  \|f^{(3)}\|_{\L_{s,3}^\infty}\sup_{\xi_1\in\R^3} \int_{(\R^3)^3}d\uxi_{2;4}'\d(\xi_1+\xi_2'-\xi_3'-\xi_4')\d(\om_{234}^1) \frac{\jp{\xi_1}^{s+\gamma}}{\jp{\xi_2'}^s\jp{\xi_3'}^s \jp{\xi_4'}^s}.
\end{equation}
Applying \cref{lem:C1} with $(\xi,\xi_1,\xi_2,\xi_3)$ replaced by $(\xi_1,\xi_2',\xi_3',\xi_4')$, we see that \eqref{eq:C1hi_bnd} holds provided that $s>2$ and $0\leq\gamma<s-2$. With this warm-up, we are now ready to prove our main proposition for $\Cr_{1,j;m}$.

\begin{prop}
\label{prop:C1hi}
For $s>2$ and $0\leq \gamma<s-2$, we have that
\begin{equation}
\|\jp{\xi_j}^{\gamma}\Cr_{1, j;m}[f^{(m+2)}]\|_{\L_{s,m}^\infty} \lesssim_{s,\gamma} \|f^{(m+2)}\|_{\L_{s,m+2}^\infty}, \qquad \forall m\in\N \text{ and } 1\leq j\leq m.
\end{equation}
\end{prop}
\begin{proof}
By symmetry of $f^{(m+2)}$ under exchange of mode labels, we may assume without loss of generality that $j=1$. For $\uxi_{2;m}\in (\R^3)^{m-1}$ fixed, let us define
\begin{equation}
g_{\uxi_{2;m}}^{(3)}(\xi_2',\xi_3',\xi_4') \coloneqq  f^{(m+2)}(\uxi_{2;m},\xi_2',\xi_3',\xi_4').
\end{equation}
Evidently, $g_{\uxi_{2;m}}^{(3)}$ is continuous on $(\R^3)^3$ by our assumption that $f^{(m+2)} \in \L_{s,m+2}^\infty$ and
\begin{align}
&\sup_{\uxi_{2;m}\in (\R^3)^{m-1}} \jp{\xi_2}^s\cdots\jp{\xi_m}^s\|g_{\uxi_{2;m}}^{(3)}\|_{\L_{s,m}^\infty} \nn\\
&\leq \sup_{(\uxi_{2;m},\uxi_{2;4}')\in (\R^3)^{m-1} \times (\R^3)^3}\jp{\xi_2'}^s\jp{\xi_3'}^s \jp{\xi_4'}^s\jp{\xi_2}^s\cdots\jp{\xi_m}^s |f^{(m+2)}(\uxi_{2;m},\uxi_{2;4}')| \nn\\
&= \|f^{(m+2)}\|_{\L_{s,m+2}^\infty}. \label{eq:gbnd}
\end{align}
With this notation, we can write
\begin{align}
&\jp{\xi_1}^{s+\gamma}\Cr_{1,1;m}[f^{(m+2)}](\uxi_{1;m})\nn \\
&= \int_{(\R^3)^3}d\uxi_{2;4}'\d(\xi_1+\xi_2'-\xi_3'-\xi_4')\d(\om_{234}^1) \frac{\jp{\xi_1}^{s+\gamma}}{\jp{\xi_2'}^s \jp{\xi_3'}^s\jp{\xi_4'}^s} \jp{\xi_2'}^s\cdots\jp{\xi_4'}^s f^{(m+2)}(\uxi_{2;m},\xi_2',\xi_3',\xi_4') \nn\\
&=\int_{(\R^3)^3}d\uxi_{2;4}'\d(\xi_1+\xi_2'-\xi_3'-\xi_4')\d(\om_{234}^1) \frac{\jp{\xi_1}^{s+\gamma}}{\jp{\xi_2'}^s \jp{\xi_3'}^s\jp{\xi_4'}^s} \jp{\xi_2'}^s\jp{\xi_3'}^s\jp{\xi_4'}^s g_{\uxi_{2;m}}^{(3)}(\uxi_{2;4}').
\end{align}
Consequently,
\begin{align}
\|\jp{\xi_1}^\gamma \Cr_{1, 1;m}[f^{(m+2)}]\|_{\L_{s,m}^\infty} &\leq \sup_{\uxi_{2;m}\in (\R^3)^{m-1}} \jp{\xi_2}^s\cdots\jp{\xi_m}^s \|g_{\uxi_{2;m}}^{(3)}\|_{\L_{s,m}^\infty} \nn\\
&\ph\times \sup_{\xi_1\in\R^3}\int_{(\R^3)^3}d\uxi_{2;4}'\d(\xi_1+\xi_2'-\xi_3'-\xi_4')\d(\om_{234}^1) \frac{\jp{\xi_1}^{s+\gamma}}{\jp{\xi_2'}^s \jp{\xi_3'}^s\jp{\xi_4'}^s} \nn\\
&\lesssim_{s,\gamma} \|f^{(m+2)}\|_{\L_{s,m+2}^\infty},
\end{align}
where we use \cref{lem:C1} and the bound \eqref{eq:gbnd} to obtain the ultimate line. Thus, the proof of the proposition is complete.
\end{proof}

\subsubsection{Boundedness of $\Cr_{2, j;m}$}
We next show the analogue of \cref{prop:C2} for the linear operator $\Cr_{2,j;m}$ defined in \eqref{eq:C2hi}. To warm up, we first consider the case $m=1$. Observe that
\begin{align}
\Cr_{2,1;1}[f^{(3)}](\xi_1) &= \int_{(\R^3)^3}d\uxi_{2;4}'\d(\xi_1+\xi_2'-\xi_3'-\xi_4')\d(\om_{234}^1) f^{(3)}(\xi_1,\xi_3',\xi_4') \nn\\
&=\int_{(\R^3)^2}d\uxi_{2;3}'\d(\om_1+\om_2' - \om_3' - \om(\xi_1+\xi_2'-\xi_3')) f^{(3)}(\xi_1,\xi_3',\xi_1+\xi_2'-\xi_3'),
\end{align}
where we have introduced the notation $\om_j'\coloneqq \om(\xi_j')$. For $\xi_1\in\R^3$ fixed, let us introduce the notation
\begin{equation}
g_{\xi_1}^{(2)}(\xi_3',\xi_4') \coloneqq f^{(3)}(\xi_1,\xi_3',\xi_4').
\end{equation}
It is straightforward to check that
\begin{equation}
\label{eq:g2bnd}
\sup_{\xi_1\in\R^3} \jp{\xi_1}^s\|g_{\xi_1}^{(2)}\|_{\L_{s,2}^\infty} \leq \|f^{(3)}\|_{\L_{s,3}^\infty}.
\end{equation}
With this notation, we can write
\begin{align}
\Cr_{2,1;1}[f^{(3)}](\xi_1) &= \int_{(\R^3)^2}d\uxi_{2;3}'\d(\om_1+\om_2' - \om_3' - \om(\xi_1+\xi_2'-\xi_3'))g_{\xi_1}^{(2)}(\xi_3',\xi_1+\xi_2'-\xi_3') \nn\\
&=\int_{(\R^3)^2}d\uxi_{2;3}'\d(\om_1+\om_2' - \om_3' - \om(\xi_1+\xi_2'-\xi_3'))\jp{\xi_3'}^{-s}\jp{\xi_1+\xi_2'-\xi_3'}^{-s} \nn\\
&\ph\jp{\xi_3'}^{s}\jp{\xi_1+\xi_2'-\xi_3'}^s g_{\xi_1}^{(2)}(\xi_3',\xi_1+\xi_2'-\xi_3').
\end{align}
By taking absolute values of both sides, it follows from this identity that for each $\xi_1\in\R^3$,
\begin{equation}
\begin{split}
\jp{\xi_1}^{s+\gamma}|\Cr_{2,1;1}[f^{(3)}](\xi_1)| &\leq \jp{\xi_1}^s\sup_{\xi_2',\xi_3'\in\R^3}\jp{\xi_3'}^{s}\jp{\xi_1+\xi_2'-\xi_3'}^s |g_{\xi_1}^{(2)}(\xi_3',\xi_1+\xi_2'-\xi_3')| \\
&\ph \times \jp{\xi_1}^\gamma \int_{(\R^3)^2}d\uxi_{2;3}'\d(\om_1+\om_2' - \om_3' - \om(\xi_1+\xi_2'-\xi_3'))\jp{\xi_3'}^{-s}\jp{\xi_1+\xi_2'-\xi_3'}^{-s}.
\end{split}
\end{equation}
Using the bound \eqref{eq:g2bnd}, we have that
\begin{equation}
\sup_{\xi_1\in\R^3} \jp{\xi_1}^s\sup_{\xi_2',\xi_3'\in\R^3}\jp{\xi_3'}^{s}\jp{\xi_1+\xi_2'-\xi_3'}^s |g_{\xi_1}^{(2)}(\xi_3',\xi_1+\xi_2'-\xi_3')| \leq \|f^{(3)}\|_{\L_{s,3}^\infty},
\end{equation}
and using \cref{lem:C2}, we have that
\begin{equation}
\sup_{\xi_1\in\R^3}\jp{\xi_1}^\gamma \int_{(\R^3)^2}d\uxi_{2;3}'\d(\om_1+\om_2' - \om_3' - \om(\xi_1+\xi_2'-\xi_3'))\jp{\xi_3'}^{-s}\jp{\xi_1+\xi_2'-\xi_3'}^{-s} < \infty.
\end{equation}
Putting together these estimates, we conclude that
\begin{equation}
\|\jp{\xi_1}^\gamma\Cr_{2,1;1}[f^{(3)}]\|_{\L_{s,1}^\infty} \lesssim_{s,\gamma}\|f^{(3)}\|_{\L_{s,3}^\infty}.
\end{equation}

Having warmed up, we are now ready to prove the generalization to the case where $m\geq 1$.
\begin{prop}
\label{prop:C2hi}
For $s>2$ and $0\leq \gamma < s-2$, we have that
\begin{equation}
\|\jp{\xi_j}^\gamma \Cr_{2, j;m}[f^{(m+2)}]\|_{\L_{s,m}^\infty} \lesssim_{s,\gamma} \|f^{(m+2)}\|_{\L_{s,m+2}^\infty}, \qquad \forall m\in\N \ \text{and} \ 1\leq j\leq m.
\end{equation}
\end{prop}
\begin{proof}
Since $\Cc_{2,1;m} = \cdots = \Cc_{2,m;m}$ and $f^{(m+2)}$ is symmetric under exchange of mode labels, it suffices to consider $j=1$. Observe that
\begin{align}
\Cr_{2,1;m}[f^{(m+2)}](\uxi_{1;m}) &= \int_{(\R^3)^3}d\uxi_{2;4}'\d(\xi_1+\xi_2'-\xi_3'-\xi_4')\d(\om_{234}^1)f^{(m+2)}(\uxi_{1;m},\xi_3',\xi_4') \nn\\
&=\int_{(\R^3)^3}d\uxi_{2;3}'\d(\om_1+\om_2'-\om_3'-\om(\xi_1+\xi_2'-\xi_3'))\jp{\xi_3'}^{-s}\jp{\xi_1+\xi_2'-\xi_3'}^{-s} \nn\\
&\ph\jp{\xi_3'}^s \jp{\xi_1+\xi_2'-\xi_3'}^{s} g_{\uxi_{1;m}}^{(2)}(\xi_3',\xi_1+\xi_2'-\xi_3'), \label{eq:C2id}
\end{align}
where for fixed $\uxi_{1;m}\in (\R^3)^m$, we have defined
\begin{equation}
g_{\uxi_{1;m}}^{(2)}(\xi_3',\xi_4') \coloneqq f^{(m+2)}(\uxi_{1;m},\xi_3',\xi_4').
\end{equation}
By the same argument as used to show \eqref{eq:g2bnd}, we have that $g_{\uxi_{1;m}}^{(2)} \in \L_{s,2}^{\infty}$ for every $\uxi_{1;m}\in (\R^3)^{m}$ and
\begin{equation}
\label{eq:gbndC2}
\sup_{\uxi_{1;m}\in (\R^3)^m} \jp{\xi_1}^s\cdots\jp{\xi_m}^s\|g_{\uxi_{1;m}}^{(2)}\|_{\L_{s,2}^\infty} \leq \|f^{(m+2)}\|_{\L_{s,m+2}^\infty}.
\end{equation}
Taking the absolute value of both sides of identity \eqref{eq:C2id}, it follows that
\begin{equation}
\begin{split}
&\jp{\xi_1}^{s+\gamma}\jp{\xi_2}^s\cdots\jp{\xi_m}^s |\Cr_{2,1;m}[f^{(m+2)}](\uxi_{1;m})|\\
&\leq \jp{\xi_1}^s\cdots\jp{\xi_m}^s \sup_{\xi_2',\xi_3'\in\R^3} \jp{\xi_3'}^s \jp{\xi_1+\xi_2'-\xi_3'}^s |g_{\uxi_{1;m}}^{(2)}(\xi_3',\xi_1+\xi_2'-\xi_3')| \\
&\ph \times \int_{(\R^3)^2}d\uxi_{2;3}'\d(\om_1+\om_2'-\om_3'-\om(\xi_1+\xi_2'-\xi_3'))\frac{\jp{\xi_1}^{\gamma}}{\jp{\xi_3'}^s\jp{\xi_1+\xi_2'-\xi_3'}^s}.
\end{split}
\end{equation}
By the bound \eqref{eq:gbndC2},
\begin{equation}
\sup_{\uxi_{1;m}\in (\R^3)^m} \jp{\xi_1}^s\cdots\jp{\xi_m}^s \sup_{\xi_2',\xi_3'\in\R^3} \jp{\xi_3'}^s \jp{\xi_1+\xi_2'-\xi_3'}^s|g_{\uxi_{1;m}}^{(2)}(\xi_3',\xi_1+\xi_2'-\xi_3')| \leq \|f^{(m+2)}\|_{\L_{s,m+2}^\infty},
\end{equation}
and by \cref{lem:C2},
\begin{equation}
\sup_{\xi_1\in \R^3} \int_{(\R^3)^2}d\uxi_{2;3}'\d(\om_1+\om_2'-\om_3'-\om(\xi_1+\xi_2'-\xi_3'))\frac{\jp{\xi_1}^{\gamma}}{\jp{\xi_3'}^s\jp{\xi_1+\xi_2'-\xi_3'}^s} < \infty.
\end{equation}
Putting together these two estimates, we arrive at the desired conclusion.
\end{proof}

\subsubsection{Boundedness of $\Cr_{3, j;m}, \Cr_{4, j;m}$}
By symmetry under swapping $\xi_{3}'$ and $\xi_4'$ and the fact that $f^{(m+2)}$ is invariant under permutation of mode labels, it suffices to only consider $\Cr_{3, j;m}$ in this sub-subsection. As before, we warm up by considering the case $m=1$. Observe that
\begin{align}
\Cr_{3, 1;1}[f^{(3)}](\xi_1) &= \int_{(\R^3)^2}d\uxi_{2;3}'\d(\om_1+\om_2'-\om_3'-\om(\xi+\xi_2'-\xi_3')) f^{(3)}(\xi_1,\xi_2',\xi_3') \nn\\
&=\int_{(\R^3)^2}d\uxi_{2;3}'\d(\om_1+\om_2'-\om_3'-\om(\xi+\xi_2'-\xi_3')) \jp{\xi_2'}^{-s}\jp{\xi_3'}^{-s} \jp{\xi_2'}^s \jp{\xi_3'}^s g_{\xi_1}^{(2)}(\xi_2',\xi_3'), \label{eq:C3rhswu}
\end{align}
where we reuse the notation $g_{\xi_1}^{(2)}$ from the last sub-subsection. Taking the absolute of both sides of \eqref{eq:C3rhswu}, we find that
\begin{equation}
\begin{split}
\jp{\xi_1}^{s+\gamma} |\Cr_{3, 1;1}[f^{(3)}](\xi_1)| &\leq \jp{\xi_1}^s\|g_{\xi_1}^{(2)}\|_{\L_{s,2}^\infty}\\
&\ph\times \jp{\xi_1}^\gamma\int_{(\R^3)^2} d\uxi_{2;3}'\d(\om_1+\om_2'-\om_3'-\om(\xi+\xi_2'-\xi_3')) \jp{\xi_2'}^{-s}\jp{\xi_3'}^{-s}.
\end{split}
\end{equation}
Combining the norm bound \eqref{eq:g2bnd} with \cref{lem:C3}, which implies that
\begin{equation}
\sup_{\xi_1\in\R^3}\jp{\xi_1}^\gamma\int_{(\R^3)^2} d\uxi_{2;3}'\d(\om_1+\om_2'-\om_3'-\om(\xi_1+\xi_2'-\xi_3')) \jp{\xi_2'}^{-s}\jp{\xi_3'}^{-s} < \infty,
\end{equation}
we conclude that
\begin{equation}
\|\jp{\xi_1}^\gamma \Cr_{3, 1;1}[f^{(3)}]\|_{\L_{s,1}^\infty} \lesssim_{s,\gamma} \|f^{(3)}\|_{\L_{s,3}^\infty}.
\end{equation}
Having warmed up, we are ready to prove the general case $m\geq 1$.

\begin{prop}
\label{prop:C3hi}
For $s>2$ and $0\leq\gamma<s-2$, we have that
\begin{equation}
\|\jp{\xi_j}^\gamma \Cr_{3, j;m}[f^{(m+2)}]\|_{\L_{s,m}^\infty} \lesssim_{s,\gamma} \|f^{(m+2)}\|_{\L_{s,m+2}^\infty}, \qquad \forall m\in\N \ \text{and} \ 1\leq j\leq m
\end{equation}
and similarly for when $\Cr_{4,j;m}$ replaces $\Cr_{3, j;m}$.
\end{prop}
\begin{proof}
Again by considerations of symmetry, we may assume without loss of generality that $j=1$. So, we need to consider the expression
\begin{align}
\Cr_{3,1;m}[f^{(m+2)}](\uxi_{1;m}) &=\int_{(\R^3)^2}d\uxi_{2;3}'\d(\om_1+\om_2'-\om_3'-\om(\xi_1+\xi_2'-\xi_3')) f^{(m+2)}(\uxi_{1;m},\xi_2',\xi_3') \nn\\
&=\int_{\R^3}d\uxi_{2;3}'\d(\om_1+\om_2'-\om_3'-\om(\xi_1+\xi_2'-\xi_3'))\jp{\xi_2'}^{-s}\jp{\xi_3'}^{-s} \jp{\xi_2'}^s\jp{\xi_3'}g_{\uxi_{1;m}}^{(2)}(\xi_2',\xi_3'). \label{eq:C3rhs}
\end{align}
Now taking the absolute value of both sides of identity \eqref{eq:C3rhs}, we find that for every $\uxi_{1;m}\in (\R^3)^m$,
\begin{equation}
|\Cr_{3, 1;m}[f^{(m+2)}](\uxi_{1;m})| \leq \|g_{\uxi_{1;m}}^{(2)}\|_{\L_{s,2}^\infty}\int_{(\R^3)^2}d\uxi_{2;3}'\d(\om_1+\om_2'-\om_3'-\om(\xi_1+\xi_2'-\xi_3'))\jp{\xi_2'}^{-s}\jp{\xi_3'}^{-s},
\end{equation}
which in turn implies that for every $\uxi_{1;m} \in (\R^3)^m$,
\begin{equation}
\label{eq:C3lhs}
\begin{split}
&\jp{\xi_1}^{s+\gamma} \jp{\xi_2}^s\cdots\jp{\xi_m}^s |\Cr_{3, 1;m}[f^{(m+2)}](\uxi_{1;m})| \\
&\leq \sup_{\uxi_{1;m}\in (\R^3)^m}\jp{\xi_1}^{s}\cdots\jp{\xi_m}^{s} \|g_{\uxi_{1;m}}^{(2)}\|_{\L_{s,2}^\infty}  \\
&\ph \times \sup_{\xi_1\in\R^3} \int_{(\R^3)^2}d\uxi_{2;3}'\d(\om_1+\om_2'-\om_3'-\om(\xi_1+\xi_2'-\xi_3'))\frac{\jp{\xi_1}^\gamma}{\jp{\xi_2'}^{s}\jp{\xi_3'}^{s}}.
\end{split}
\end{equation}
The desired conclusion now follows from combining the norm bound \eqref{eq:gbndC2} with \cref{lem:C3} and taking the supremum of the left-hand side of inequality \eqref{eq:C3lhs} over $\uxi_{1;m}\in (\R^3)^m$.
\end{proof}

\subsubsection{Boundedness of $\Cr$}
Given a hierarchy $F=(f^{(m)})_{m=1}^\infty$, the decompositions \eqref{eq:Crdcomp} and \eqref{eq:Crjdcomp} yield
\begin{equation}
\Cr[F]^{(m)} = \sum_{j=1}^m \paren*{\Cr_{1, j;m}[f^{(m+2)}] + \Cr_{2, j;m}[f^{(m+2)}] - \Cr_{3, j;m}[f^{(m+2)}] - \Cr_{4, j;m}[f^{(m+2)}]}.
\end{equation}
By the triangle inequality and \cref{prop:C1hi}, \cref{prop:C2hi}, \cref{prop:C3hi} applied with $\gamma=0$, we then arrive at the following lemma.

\begin{lemma}
\label{lem:Cjm}
For $s>2$, we have that for every $m\in\N$,
\begin{equation}
\sup_{1\leq j\leq m} \|\Cr_{j;m}[f^{(m+2)}]\|_{\L_{s,m}^\infty} \lesssim_s \|f^{(m+2)}\|_{\L_{s,m+2}^\infty}.
\end{equation}
\end{lemma}

With another application of the triangle inequality, \cref{lem:Cjm} implies the estimate
\begin{equation}
\|\Cr[F]^{(m)}\|_{\L_{s,m}^\infty} \lesssim_{s} m\|f^{(m+2)}\|_{\L_{s,m+2}^\infty}.
\end{equation}
Now suppose that $0<\ep'<\ep < 1$. Then
\begin{equation}
(\ep')^m\|\Cr[F]^{(m)}\|_{\L_{s,m}^\infty} \leq C_{s} (\ep')^m m\|f^{(m+2)}\|_{\L_{s,m+2}^\infty} \leq \frac{C_{s}}{\ep^2} (\frac{\ep'}{\ep})^m m \ep^{m+2}\|f^{(m+2)}\|_{\L_{s,m+2}^\infty},
\end{equation}
where $C_{s}>0$ is a constant depending only on the data $s$. Summing over the infinite range $m\geq 1$ and recalling the definition of the $\Lr_{s,\ep}^\infty$ norm from \eqref{eq:hnorm}, we arrive at the following proposition.

\begin{prop}
\label{prop:C}
For $s>2$, we have that
\begin{equation}
\label{eq:Closs}
\|\Cr[F]\|_{\Lr_{s,\ep'}^\infty} \lesssim_s \ep^{-2}\paren*{\sup_{m\in\N} m(\ep'\ep^{-1})^m} \|F\|_{\Lr_{s,\ep}^\infty}, \qquad \forall  0< \ep' < \ep < 1.
\end{equation}
\end{prop}

We emphasize that \eqref{eq:Closs} is a losing estimate in the sense that we have decreased the parameter $\ep'$ in the norm on the left-hand side compared to the parameter $\ep$ in the norm on the right-hand side in order to compensate for the linear growth of the $m$ factor. Overcoming this loss will be a key issue as we attempt to prove well-posedness of the spectral hierarchy in the next section.

\section{Well-posedness for WKE hierarchy}
\label{sec:WP}
We now have all the ingredients necessary to show the well-posedness of the spectral hierarchy \eqref{eq:WKE_hier}, ultimately proving \cref{thm:main}. As commented in \cref{ssec:intropf} of the introduction, a key ingredient to our proof is the iterated Duhamel expansion. We write the spectral hierarchy \eqref{eq:WKE_hier} in integral form
\begin{equation}
\label{eq:Duh}
F(t) = F_0 + \int_0^t \Cr[F(\tau)]d\tau,
\end{equation}
where $F=(f^{(m)})_{m=1}^\infty$ and $F_0 = (f_0^{(m)})_{m=1}^\infty$. Now a few observations are in order. First, if $F\in C([0,T]; \Lr_{s,\ep}^\infty)$, for some $s>2$ and $0<\ep<1$, satisfies the equation \eqref{eq:Duh}, then since by \cref{prop:C}
\begin{equation}
\|\Cr[F(\tau)]\|_{\Lr_{s,\ep'}^\infty} \lesssim_{s,\ep,\ep'} \|F(\tau)\|_{\Lr_{s,\ep}^\infty}, \qquad \forall \tau\in [0,T],
\end{equation}
for any $0<\ep'<\ep<1$, the fundamental theorem of calculus implies that $F\in C^1([0,T]; \Lr_{s,\ep'}^\infty)$.  Consequently, given any $k\in\N$, we have that $F\in C^k([0,T];\Lr_{s,\ep'}^\infty)$. Second, under the same assumptions on $F$, the collision operator $\Cr$ commutes with integration in time. Indeed, this follows readily from Fubini-Tonelli. Thus, we can obtain the \emph{iterated Duhamel expansion}
\begin{equation}
\label{eq:Duhexp}
\begin{split}
F(t) &=\sum_{k=0}^{j} \frac{t^k}{k!}\Cr^k[F_0] + \int_0^t\int_0^{t_1}\cdots\int_0^{t_j}\Cr^{j+1}[F(t_{j+1})]dt_{j+1}\cdots dt_2dt_1.
\end{split}
\end{equation}
To ease the burden of notation, we define the \emph{Duhamel iterates}
\begin{equation}
\label{eq:Duhit}
\Du_j(F_0,t) \coloneqq  \sum_{k=0}^j\frac{t^k}{k!}\Cr^k[F_0],
\end{equation}
so that with this notation, equation \eqref{eq:Duhexp} becomes
\begin{equation}
\label{eq:Duhexpnot}
F(t_0) = \Du_j(F_0,t_0) + \int_0^{t_0}\cdots \int_0^{t_j}\Cr^{j+1}[F(t_{j+1})]dt_{j+1}\cdots dt_{1}, \qquad \forall j\in\N_0.
\end{equation}
Unpacking the definition of $\Cr^j$ yields the component-wise formula
\begin{equation}
\label{eq:Duhitm}
\Du_j(F_0,t)^{(m)} = \sum_{k=0}^j\frac{t^k}{k!}\sum_{\ul{\mu}_k \in \A_k^{(m)}} \Cr_{\mu_1;m}\cdots\Cr_{\mu_k; m+2(k-1)}[f_0^{(m+2k)}],
\end{equation}
where the summation is over all tuples $\ul{\mu}_k=(\mu_1,\ldots,\mu_k)$ belonging to the set
\begin{equation}
\label{eq:Adef}
\A_k^{(m)} \coloneqq \{\umu_k = (\mu_1,\ldots,\mu_k) : 1\leq \mu_r\leq m+2(r-1) \enspace \forall 1\leq r\leq k\}.
\end{equation}

\begin{remark}
\label{rem:Acard}
For later use, we note that the cardinality of the set $\A_k^{(m)}$ is
\begin{equation}
\prod_{r=1}^k (m+2r-2),
\end{equation}
which for fixed $k$, is of size $O(m^k)$ as $m\rightarrow\infty$, and for fixed $m$, is of size $O( (m+2k)!)$ as $k\rightarrow\infty$.
\end{remark}

\subsection{Convergence of Duhamel series}
\label{ssec:WPcon}
The goal of this subsection is to prove the convergence of the Duhamel series \eqref{eq:Duhs}. The main result is the following proposition.
\begin{prop}
\label{prop:Duhcon}
Let $F_0= (f_0^{(m)})_{m=1}^\infty \in \Lr_{s,\ep_1}^\infty$ for some $s>2$ and $\ep_1>0$. There exists a constant $C_s>0$ such that for any parameters $\ep_2, T>0$ satisfying $\ep_2e^{TC_s\ep_1^{-2}}\ep_1^{-1} < 1$ and $C_s T\ep_1^{-2} < 1$, the series
\begin{equation}
F(t)\coloneqq \sum_{k=0}^\infty \frac{t^k}{k!}\Cr^k[F_0]
\end{equation}
converges absolutely in $C([-T,T]; \Lr_{s,\ep_2}^\infty)$ and is a solution to equation \eqref{eq:WKE_hier} on $[0,T]$. Moreover,
\begin{equation}
\sup_{0\leq |t|\leq T} \|F(t)\|_{\Lr_{s,\ep_2}^\infty} < \paren*{\frac{\ep_2e^{TC_s\ep_1^{-2}}\ep_1^{-1}}{1-\ep_2 e^{TC_s\ep_1^{-2}}\ep_1^{-1}}  + \frac{(C_sT\ep_1^{-2})^2}{(1-C_sT\ep_1^{-2})} }\|F_0\|_{\Lr_{s,\ep_1}^\infty}.
\end{equation}
\end{prop}

We prove \cref{prop:Duhcon} through a series of lemmas, starting with a bound for the $m$-mode component of collision iterates.
\begin{lemma}
\label{lem:Duhit}
For $s>2$, there exists a constant $C_s>0$ such that
\begin{equation}
\label{eq:Duhitt}
\|\Cr_{\mu_1;m}\cdots\Cr_{\mu_j; m+2j-2}[f_0^{(m+2j)}]\|_{\L_{s,m}^\infty} \leq C_s^{j} \|f_0^{(m+2j)}\|_{\L_{s,m+2j}^\infty}
\end{equation}
for every $m\in\N$, $j\in\N$, $\umu_j\in\A_j^{(m)}$. Consequently,
\begin{equation}
\label{eq:Cjmbnd}
\|\Cr^j[F_0]^{(m)}\|_{\L_{s,m}^\infty} \leq C_s^j \paren*{\prod_{r=1}^j m+2r-2} \|f_0^{(m+2j)}\|_{\L_{s,m+2j}^\infty}.
\end{equation}
\end{lemma}
\begin{proof}
The first bound follows from applying \cref{lem:Cjm} a total of $j$ times. Now unpacking the definition \eqref{eq:Chidef} of $\Cr$ and then expanding the $j$-fold contraction operator $\Cr^j$, we find that
\begin{equation}
\Cr^j[F_0]^{(m)} = \sum_{\umu_j \in \A_j^{(m)}} \Cr_{\mu_1;m}\cdots\Cr_{\mu_j;m+2j-2}[f_0^{(m+2j)}],
\end{equation}
where the reader will recall the definition of the tuple set $\A_j^{(m)}$ from \eqref{eq:Adef}. Using the triangle inequality and estimate \eqref{eq:Duhitt} gives
\begin{align}
\|(\Cr^j[F_0])^{(m)}\|_{\L_{s,m}^\infty} \leq C_s^j |\A_j^{(m)}| \|f_0^{(m+2j)}\|_{\L_{s,m+2j}^\infty}\leq C_s^j \paren*{\prod_{r=1}^j m+2r-2} \|f_0^{(m+2j)}\|_{\L_{s,m+2j}^\infty}, 
\end{align}
where the ultimate follows from \cref{rem:Acard}.
\end{proof}

Next, recalling from \eqref{eq:Duhit} the definition of $\Du_j(F_0,t)$, we see that for any $j,j'\geq 0$
\begin{align}
\|\Du_{j'}(F_0,t)^{(m)} - \Du_j(F_0,t)^{(m)}\|_{\L_{s,m}^\infty} &\leq \sum_{k=\min\{j,j'\}}^{\max\{j,j'\}} \frac{t^k}{k!} \|\Cr^k[F_0]^{(m)}\|_{\L_{s,m}^\infty} \nn\\
&\leq \sum_{k=\min\{j,j'\}}^{\max\{j,j'\}} \frac{C_s^k t^k}{k!} \paren*{\prod_{r=1}^k m+2r-2}\|f_0^{(m+2k)}\|_{\L_{s,m+2k}^\infty}, \label{eq:Dujm_bnd}
\end{align}
where the first line follows from the triangle inequality and the second line from \cref{lem:Duhit}. Let $0<\ep_2<\ep_1 < 1$, where $F_0\in \Lr_{s,\ep_1}^\infty$ as in the statement of \cref{prop:Duhcon} and $\ep_2$ is to be determined. We observe from the estimate \eqref{eq:Dujm_bnd} and remembering the definition \eqref{eq:hnorm} of the norm $\Lr_{s,\ep_2}^\infty$ that
\begin{align}
\|\Du_{j'}(F_0,t) -\Du_j(F_0,t)\|_{\Lr_{s,\ep_2}^\infty} &\leq \sum_{m=1}^\infty \ep_2^m\sum_{k=\min\{j,j'\}}^{\max\{j,j'\}} \frac{C_s^k t^k}{k!} \paren*{\prod_{r=1}^k m+2r-2} \|f_0^{(m+2k)}\|_{\L_{s,m+2k}^\infty}. \label{eq:Duhmaj}
\end{align}
Thus, if we can show that upon appropriately choosing $\ep_2$ and $t$, the preceding expression is finite for $j=0$ and $j'=\infty$, then we will have shown the convergence assertion in \cref{prop:Duhcon}.

\begin{lemma}
Let $s>2$. There exists a constant $C_s>0$, such that for $F_0=(f_0^{(m)})_{m=1}^\infty \in \Lr_{s,\ep_1}^\infty$ with $\ep_1>0$,
\begin{equation}
\begin{split}
&\sup_{|t|\leq T} \sum_{m=1}^\infty \ep_2^m\sum_{k=0}^\infty \frac{C_s^k t^k}{k!} \paren*{\prod_{r=1}^k m+2r-2} \|f_0^{(m+2k)}\|_{\L_{s,m+2k}^\infty} \\
&< \frac{\ep_2e^{TC_s\ep_1^{-2}}\ep_1^{-1}}{1-\ep_2 e^{TC_s\ep_1^{-2}}\ep_1^{-1}} \|F_0\|_{\Lr_{s,\ep_1}^\infty} + \frac{(C_sT\ep_1^{-2})^2}{(1-C_sT\ep_1^{-2})} \|F_0\|_{\Lr_{s,\ep_1}^\infty},
\end{split}
\end{equation}
provided that  $\ep_2e^{TC_s\ep_1^{-2}}\ep_1^{-1} < 1$ and $C_s T\ep_1^{-2} < 1$.
\end{lemma}
\begin{proof}
The idea of the proof is to divide the double summation over $k$ and $m$ into two pieces: $k\leq m$ and $k>m$. The first piece is relatively easy, as for fixed $m$, the summation over $k$ only grows exponentially in $m$. We can absorb this growth by choosing $\ep_2$ sufficiently small. The second piece is more difficult, as for fixed $k$, the summation over $m$ grows at most like $k^k$. To absorb this growth, we have to use the factorial denominator together with the remaining freedom to restrict time to an arbitrarily small interval. Without loss of generality, we may assume that $t\geq 0$.

For $k\leq m$, write
\begin{equation}
\label{eq:prod_mk}
\prod_{r=1}^k (m+2r-2)= m^k \prod_{r=1}^k (1+\frac{2r-2}{m}) \leq (3m)^k.
\end{equation}
Hence, we have that
\begin{align}
\sum_{k=0}^m \frac{C_s^k t^k\prod_{r=1}^k (m+2r-2)}{k!\ep_1^{2k}} &\leq \sum_{k=0}^m \frac{(tC_s'm\ep_1^{-2})^k}{k!} \leq \exp(tC_s'm \ep_1^{-2}) = \paren*{\exp(tC_s'\ep_1^{-2})}^{m},
\end{align}
where $C_s'>C_s$. This inequality implies that
\begin{align}
&\sum_{m=1}^\infty (\frac{\ep_2}{\ep_1})^{m}\sum_{k=0}^m\frac{C_s^k t^k\prod_{r=1}^k (m+2r-2)}{k!\ep_1^{2k}} \ep_1^{m+2k}\|f_0^{(m+2k)}\|_{\L_{s,m+2k}^\infty} \nn\\
&\leq \paren*{\sup_{m\in\N} \ep_1^{m+2k}\|f_0^{(m+2k)}\|_{\L_{s,m+2k}^\infty} } \sum_{m=1}^\infty \paren*{\ep_2e^{tC_s'\ep_1^{-2}}\ep_1^{-1}}^m \nn\\
&\leq \frac{\ep_2e^{tC_s'\ep_1^{-2}}\ep_1^{-1}}{1-\ep_2 e^{tC_s'\ep_1^{-2}}\ep_1^{-1}} \|F_0\|_{\Lr_{s,\ep_1}^\infty}, \label{eq:lowk}
\end{align}
provided that $\ep_2$ is sufficiently small so that
\begin{equation}
\ep_2e^{tC_s'\ep_1^{-2}}\ep_1^{-1} < 1.
\end{equation}

For $k>m$, we have the crude inequality
\begin{equation}
\label{eq:prod_km}
\prod_{r=1}^k (m+2r-2) < \prod_{r=1}^k (k+2r-2) < (3k)^k.
\end{equation}
Using Robbins' factorial bounds \cite{Robbins1955}
\begin{equation}
\label{eq:Robb}
\sqrt{2\pi}j^{j+\frac{1}{2}}e^{-j}e^{\frac{1}{12j+1}} \leq j! \leq \sqrt{2\pi}j^{j+\frac{1}{2}}e^{-j}e^{\frac{1}{12j}}, \qquad \forall j\in\N,
\end{equation}
it follows now that
\begin{equation}
\frac{\prod_{r=1}^k (m+2r-2)}{k!} < \frac{(3ek)^k}{\sqrt{2\pi}k^{k+\frac{1}{2}}e^{\frac{1}{12k+1}}} < \frac{(3e)^k}{\sqrt{2\pi k}}.
\end{equation}
From this bound, we obtain that there is a constant $C_s''>C_s$, such that
\begin{align}
\frac{C_s^k t^k\prod_{r=1}^k(m+2r-2)}{k!} &< (2\pi k)^{-1/2}\paren*{{C_s'' t}}^k.
\end{align}
Hence,
\begin{align}
&\sum_{m=1}^\infty \ep_2^m\sum_{k=m+1}^\infty \frac{C_s^kt^k\prod_{r=1}^k (m+2r-2)}{k!} \|f_0^{(m+2k)}\|_{\L_{s,m+2k}^\infty} \nn\\
&<\sum_{m=1}^\infty \frac{(C_s'' t\ep_1^{-2})^{m+1}}{(2\pi m)^{1/2}} \sum_{k=m+1}^\infty \ep_1^{m+2k}\|f_0^{(m+2k)}\|_{\L_{s,m+2k}^\infty} \nn\\
&\leq \|F_0\|_{\Lr_{s,\ep_1}^\infty} \sum_{m=1}^\infty \frac{(C_s'' t\ep_1^{-2})^{m+1}}{(2\pi m)^{1/2}} \nn\\
&< \frac{(C_s''t\ep_1^{-2})^2}{(1-C_s''t\ep_1^{-2})} \|F_0\|_{\Lr_{s,\ep_1}^\infty}, \label{eq:hik}
\end{align}
provided that $C_s''t\ep_1^{-2} < 1$.

Putting together the estimates \eqref{eq:lowk} and \eqref{eq:hik}, we have shown that
\begin{equation}
\label{eq:Dubnd}
\begin{split}
&\sum_{m=1}^\infty \ep_2^m\sum_{k=0}^\infty \frac{C_s^k t^k}{k!}\paren*{\prod_{r=1}^k(m+2k-2)} \|f_0^{(m+2k)}\|_{\Lr_{s,m+2k}^\infty} \\
&< \frac{\ep_2e^{tC_s'\ep_1^{-2}}\ep_1^{-1}}{1-\ep_2 e^{tC_s'\ep_1^{-2}}\ep_1^{-1}} \|F_0\|_{\Lr_{s,\ep_1}^\infty} + \frac{(C_s''t\ep_1^{-2})^2}{(1-C_s''t\ep_1^{-2})} \|F_0\|_{\Lr_{s,\ep_1}^\infty},
\end{split}
\end{equation}
provided that 
\begin{equation}
\label{eq:bal}
\ep_2e^{tC_s'\ep_1^{-2}}\ep_1^{-1} < 1 \quad \text{ and } \quad C_s''t\ep_1^{-2} < 1.
\end{equation}
To see that we can satisfy this constraint, recall that $\ep_1$ is part of the initial data and therefore fixed. We assume that $T>0$ is sufficiently small so that $\max\{C_s', C_s''\} T\ep_1^{-2}\leq 1/2$. We then choose $\ep_2<\ep_1 e^{-1/2}$.
\end{proof}

Lastly, we check that the absolutely convergent series
\begin{equation}
\lim_{j\rightarrow\infty}\Du_j(F_0, t) = \sum_{k=0}^\infty \frac{t^k}{k!}\Cr^k[F_0]
\end{equation}
indeed defines a solution to the spectral hierarchy \eqref{eq:WKE_hier}. To this end, let us denote the right-hand side above by $F(t)$. Recalling the identity \eqref{eq:Duhitm}, we have shown above that
\begin{equation}
F(t)^{(m)} = \sum_{k=0}^\infty\frac{t^k}{k!}\sum_{\umu_k\in\A_k^{(m)}} \Cr_{\mu_1;m}\cdots\Cr_{\mu_k; m+2(k-1)}[f_0^{(m+2k)}]
\end{equation}
converges uniformly in $\L_{s,m}^\infty$ on the interval $[-T,T]$. Therefore, we may differentiate with respect to time inside the summation to obtain that
\begin{align}
\p_t F(t)^{(m)} &= \sum_{j=0}^\infty \frac{t^j}{j!} \sum_{\umu_{j+1}\in\A_{j+1}^{(m)}} \Cr_{\mu_1;m}\cdots\Cr_{\mu_{j+1}; m+2j}[f_0^{(m+2j+2)}], \label{eq:inser}
\end{align}
where we have made the change of variable $j=k-1$. Unpacking the definition \eqref{eq:Chidef}, we have that
\begin{align}
\Cr^j[F_0]^{(m+2)} = \sum_{\umu_j' \in \A_{j}^{(m+2)}} \Cr_{\mu_1';m+2}\cdots\Cr_{\mu_j'; m+2j}[f_0^{(m+2j+2)}].
\end{align}
Setting $\mu_{i}' = \mu_{i+1}$ for $1\leq i\leq j$, we obtain
\begin{align}
\sum_{\umu_{j+1}\in\A_{j+1}^{(m)}} \Cr_{\mu_1;m}\cdots\Cr_{\mu_{j+1}; m+2j}[f_0^{(m+2j+2)}] &= \sum_{\mu_1=1}^m \sum_{\umu_j'\in\A_j^{(m+2)}} \Cr_{\mu_1;m}\Cr_{\mu_1';m+2}\cdots\Cr_{\mu_j';m+2j+2}[f_0^{(m+2j+2)}] \nn\\
&=\sum_{\mu_1=1}^m \Cr_{\mu_1;m}[\Cr^j[F_0]^{(m+2)}],
\end{align}
where we also use the linearity of $\Cr_{\mu_1;m}$ to obtain the ultimate line. Inserting this identity into equation \eqref{eq:inser} and using the linearity and continuity of $\Cr_{\mu_1;m}$ from the space $\L_{s,m+2}^\infty$ to the space $\L_{s,m}^\infty$ to move the $j$-summation inside the argument of $\Cr_{\mu_1;m}$, we arrive at
\begin{equation}
\p_t F(t)^{(m)} = \sum_{\mu_1=1}^m \Cr_{\mu_1;m}\left[\sum_{j=0}^\infty \frac{t^j}{j!} \Cr^j[F_0]^{(m+2)}\right] = \sum_{\mu_1=1}^m \Cr_{\mu_1;m}[F(t)^{(m+2)}].
\end{equation}
With this last step, we conclude the proof of \cref{prop:Duhcon}.

\subsection{Uniqueness}
\label{ssec:WPu}
Next, we show that for $s>2$ and $F_0\in \Lr_{s,\ep_1}^\infty$, the absolutely convergent Duhamel series is the unique solution in the class $C([0,T]; \Lr_{s,\ep_2}^\infty)$ to the spectral hierarchy equation \eqref{eq:WKE_hier}.

\begin{prop}
\label{prop:uniq}
For $s>2$ and $T>0$, suppose that $F,G\in C([0,T];\Lr_{s,\ep}^\infty)$ satisfy equation \eqref{eq:Duh}. Then $F=G$.
\end{prop}

Since equation \eqref{eq:Duh} is linear in $F$, the proof of uniqueness reduces to showing that the zero solution is the unique solution starting from zero initial datum. We recall that any solution $F\in C([-T,T];\Lr_{s,\ep}^\infty)$ can be written as
\begin{equation}
F(t_0) = \Du_j(F_0,t_0) + \int_0^{t_0}\int_{0}^{t_1}\cdots\int_{0}^{t_j} \Cr^{j+1}[F(t_{j+1})]dt_{j+1}\cdots dt_2dt_1
\end{equation}
for any $j\geq 0$. Consequently, if $F_0 = 0$, then $\Du_j(F_0,\cdot)\equiv 0$, so that
\begin{equation}
F(t_0) = \int_0^{t_0}\int_{0}^{t_1}\cdots\int_{0}^{t_j} \Cr^{j+1}[F(t_{j+1})]dt_{j+1}\cdots dt_2dt_1.
\end{equation}
At the level of the hierarchy components, the preceding identity becomes
\begin{equation}
f^{(m)}(t_0) = \int_0^{t_0}\int_{0}^{t_1}\cdots\int_{0}^{t_j} \Cr^{j+1}[F(t_{j+1})]^{(m)}dt_{j+1}\cdots dt_2dt_1.
\end{equation}
Since uniqueness is a local property, it only remains for us to show that the term in the right-hand side vanishes as $j\rightarrow\infty$.

\begin{lemma}
\label{lem:Duherr}
If $F\in C([0,T];\Lr_{s,\ep}^\infty)$ is a solution to equation \eqref{eq:Duh}, for $s>2$ and $0<\ep<1$, and $T=T(s,\ep)>0$ is sufficiently small, then for every $m\in\N$,
\begin{equation}
\lim_{j\rightarrow\infty} \sup_{0\leq t_0\leq T} \int_0^{t_0}\int_{0}^{t_1}\cdots\int_{0}^{t_j} \|\Cr^{j+1}[F(t_{j+1})]^{(m)}\|_{\L_{s,m}^\infty}dt_{j+1}\cdots dt_2dt_1 = 0.
\end{equation}
\end{lemma}
\begin{proof}
Applying the estimate \eqref{eq:Cjmbnd} with $F_0$ replaced by $F(t_{j+1})$, we see that
\begin{equation}
\|\Cr^{j+1}[F(t_{j+1})]^{(m)}\|_{\L_{s,m}^\infty} \leq C_s^{j+1}\paren*{\prod_{r=1}^{j+1}(m+2r-2)} \|f^{(m+2j+2)}(t_{j+1})\|_{\L_{s,m+2j+2}^\infty}.\end{equation}

If $j\geq m$, we can use the bound \eqref{eq:prod_km} to obtain
\begin{align}
C_s^{j+1}\paren*{\prod_{r=1}^{j+1}(m+2r-2)} \|f^{(m+2j+2)}(t_{j+1})\|_{\L_{s,m+2j+2}^\infty} &\leq C_s^{j+1}(3(j+1))^{j+1}\|f^{(m+2j+2)}(t_{j+1})\|_{\L_{s,m+2j+2}^\infty} \nn\\
&\leq \frac{C_s'^{j+1}(j+1)^{j+1}}{\ep^{m+2j+2}} \|F\|_{C([0,T]; \Lr_{s,\ep}^\infty)},
\end{align}
where $C_s'\geq C_s$. Hence,
\begin{equation}
\begin{split}
&\sup_{0\leq t_0\leq T}\int_0^{t_0}\int_{0}^{t_1}\cdots\int_{0}^{t_j} \|f^{(m+2j+2)}(t_{j+1})\|_{\L_{s,m+2j+2}^\infty}dt_{j+1}\cdots dt_2dt_1 \\
&\leq \frac{C_s'^{j+1}T^{j+1}(j+1)^{j+1}}{(j+1)!\ep^{m+2j+2}} \|F\|_{C([0,T]; \Lr_{s,\ep}^\infty)}.
\end{split}
\end{equation}
Next, using Robbins' lower bound in \eqref{eq:Robb}, we see that
\begin{equation}
\frac{C_s'^{j+1}T^{j+1}(j+1)^{j+1}}{(j+1)!\ep^{2j+2}} < (j+1)^{-1/2}\paren*{\frac{C_s' T}{\ep^2}}^{j+1},
\end{equation}
which tends to zero as $j\rightarrow\infty$, provided that $T$ is sufficiently small so that $C_s'T\ep^{-2}\leq 1$. Thus,
\begin{equation}
\lim_{j\rightarrow \infty} \sup_{0\leq t_0\leq T} \int_0^{t_0}\int_{0}^{t_1}\cdots\int_{0}^{t_j} \|f^{(m+2j+2)}(t_{j+1})\|_{\L_{s,m+2j+2}^\infty}dt_{j+1}\cdots dt_2dt_1 = 0,
\end{equation}
completing the proof of the lemma.
\end{proof}

\subsection{Dependence on initial data}
\label{ssec:WPdep}
Finally, we quantify the dependence of the Duhamel series on the datum $F_0$. Since the Duhamel series, and by implication the solution to the spectral hierarchy \eqref{eq:WKE_hier}, is linear in this dependence, our task reduces to an application of the estimates from \cref{ssec:WPcon}.

\begin{prop}
\label{prop:dep}
Let $F_0=(f_0^{(m)})_{m=1}^\infty$ and $G_0=(g_0^{(m)})_{m=1}^\infty$, let $F$ and $G$ be their respective solutions in $C([0,T]; \Lr_{s,\ep_1}^\infty)$ starting from $F_0$ and $G_0$. Then there exist $\ep_2(s,\ep_1)>0$ and $0<T'(s,\ep_1,\ep_2)\leq T$ such that
\begin{equation}
\sup_{0\leq t\leq T'} \|F(t)-G(t)\|_{\Lr_{s,\ep_2}^\infty} \lesssim_{s,T,\ep_1,\ep_2} \|F_0-G_0\|_{\Lr_{s,\ep_1}^\infty}.
\end{equation}
\end{prop}
\begin{proof}
Let us first set $H_0\coloneqq F_0-G_0$, and define
\begin{equation}
H(t) \coloneqq \sum_{k=0}^\infty \frac{t^k}{k!}\Cr^k[H_0].
\end{equation}
\cref{prop:Duhcon}, specifically inequality \eqref{eq:Dubnd} shows that there is a constant $C_s>0$ such that
\begin{equation}
\|H\|_{C([-T',T']; \Lr_{s,\ep_2}^\infty)} < \frac{\ep_2e^{T'C_s\ep_1^{-2}}\ep_1^{-1}}{1-\ep_2 e^{T'C_s\ep_1^{-2}}\ep_1^{-1}} \|H_0\|_{\Lr_{s,\ep_1}^\infty} + \frac{(C_s T'\ep_1^{-2})^2}{(1-C_s T'\ep_1^{-2})} \|H_0\|_{\Lr_{s,\ep_2}^\infty}
\end{equation}
provided that $\ep_1,\ep_2,T'$ are chosen to satisfy the constraints \eqref{eq:bal}. Moreover, $H$ is a solution to the spectral hierarchy with initial datum $H_0$. Since $F-G$ is also solution to the spectral hierarchy with initial datum $H_0$ on the interval $[0,T']$, \cref{prop:uniq} implies that $F-G=H$ on this interval.
\end{proof}

\section{Super-statistical solutions}
\label{sec:SS}
We conclude this article with some remarks on super-statistical solutions, briefly mentioned at the end of \cref{ssec:intromr}. To start off, we note that the spectral correlation functions $\{f_{L,\lambda}^{(m)}\}_{m=1}^\infty$ (recall their definition from \eqref{eq:cordef}) satisfy the \emph{admissibility} or \emph{consistency} condition
\begin{equation}
\int_{\R^3}d\xi_{m+1} f_{L,\lambda}^{(m+1)}(\xi_1,\ldots,\xi_{m+1}) = Cf_{L,\lambda}^{(m)}(\xi_1,\ldots,\xi_m),
\end{equation}
where $C>0$ is some fixed constant. Indeed, by Fubini-Tonelli, the definition of the empirical spectrum \eqref{eq:esdef}, and conservation of mass for the NLS, we have that
\begin{align}
\int_{\R^3}d\xi_{m+1} f_{L,\lambda}^{(m+1)}(\xi_1,\ldots,\xi_{m+1}) &= \E\paren*{\int_{\R^3}d\xi_{m+1} f_{L,\lambda}^{\otimes m+1}(\xi_1,\ldots,\xi_{m+1})} \nn\\
&=\E\paren*{ f_{L,\lambda}^{\otimes m}(\xi_1,\ldots,\xi_m) } L^{-d}\sum_{k\in\Z_L^d} \phi(k),
\end{align}
where $\phi$ is the same deterministic function from the initial data formula \eqref{eq:ID}. So in the kinetic limit, we expect that--but at present cannot prove--solutions to the spectral hierarchy satisfy the condition
\begin{equation}
\label{eq:adm}
\int_{\R^3}d\xi_{m+1}f^{(m+1)}(\xi_1,\ldots,\xi_{m+1}) = \paren*{\int_{\R^3}\phi(k)dk} f^{(m)}(\xi_1,\ldots,\xi_{m}).
\end{equation}
Without loss of generality, we may assume that the prefactor equals one, so that $f^{(m)}$ has unit mass.

We argue now silimiarly as to in \cite[Section 5]{Spohn1981}. Given data $F_0 = (f_0^{(m)})_{m=1}^\infty \in \Lr_{s,\ep}^\infty$ such that $f_0^{(m)}\geq 0$, is symmetric under permutation, and satisfies the admissibility condition \eqref{eq:adm}, the Hewitt-Savage/de Finetti theorem \cite{HS1955}, implies that there exists a probability measure $\rho$ on the set $\P(\R^3)$ of Borel probability measures on $\R^3$, such that
\begin{equation}
f_0^{(m)} = \int_{\P(\R^3)}d\rho(f_0) f_0^{\otimes m},
\end{equation}
with equality in the sense of measures. For any nonnegative test function $\psi\in\Sc(\R^3)$,
\begin{equation}
\int_{(\R^3)^m}df_0^{\otimes m}(\uxi_{1;m})\psi^{\otimes m}(\uxi_{1;m}) = \paren*{\int_{\R^3}df_0(\xi)\psi(\xi)}^{m},
\end{equation}
which implies that
\begin{align}
0\leq \int_{\P(\R^3)}d\rho(f_0)\paren*{\int_{\R^3}df_0(\xi)\psi(\xi)}^{m} &=\int_{(\R^3)^m}d\uxi_{1;m}f_0^{(m)}(\uxi_{1;m})\psi^{\otimes m}(\uxi_{1;m}) \nn\\
&\leq \frac{\|F_0\|_{\Lr_{s,\ep}^\infty}}{\ep^m} \paren*{\int_{\R^3}d\xi \jp{\xi}^{-s}\psi(\xi)}^m
\end{align}
for every $m\in\N$, where we use our assumption that $F_0=(f_0^{(m)})_{m=1}^\infty \in \Lr_{s,\ep}^\infty$. Hence, for any $R>0$, Chebyshev's inequality implies that
\begin{align}
\rho\paren*{\int_{\R^3}df_0(\xi)\psi(\xi) \geq R} \leq \frac{\|F_0\|_{\Lr_{s,\ep}^\infty}}{(\ep R)^m}\paren*{\int_{\R^3}d\xi \jp{\xi}^{-s}\psi(\xi)}^m, \qquad \forall m\in\N.
\end{align}
Choosing
\begin{equation}
R> \ep^{-1}\int_{\R^3}d\xi \jp{\xi}^{-s}\psi(\xi)
\end{equation}
and letting $m\rightarrow\infty$, we conclude that with $\rho$-probability one,
\begin{equation}
\int_{\R^3}df_0(\xi)\psi(\xi)  \leq \ep^{-1}\int_{\R^3}d\xi \jp{\xi}^{-s}\psi(\xi), \qquad \forall \psi\in\Sc(\R^3), \ \psi\geq 0.
\end{equation}
Implicitly, we have used the fact that $\Sc(\R^3)$ is a separable space and that the countable union of $\rho$-null sets is again $\rho$-null. Since $\Sc(\R^3)$ is dense in the space,
\begin{equation}
L_{-s}^1 \coloneqq \{\psi \in \Sc'(\R^3) : \|\jp{\xi}^{-s}\psi\|_{L^1} < \infty\},
\end{equation}
it follows that with $\rho$-probability one, $f_0$ defines a positive, bounded linear functional on this space. Consequently, the Riesz representation theorem implies that $f_0$ can be identified with a nonnegative function in $L^\infty(\R^3)$ such that $\|\jp{\xi}^sf_0\|_{L^\infty} < \infty$.

Now the above only shows that for $f_0\in\supp(\rho)$, $\|\jp{\xi}^sf_0\|_{L^\infty}<\infty$ $\rho$-almost surely; we do not know that $f_0$ is continuous on $\R^3$. For the sake of argument, though, suppose that
\begin{equation}
\supp(\rho)\subset \{f_0\in L_s^\infty(\R^3): \|f_0\|_{L_s^\infty} \leq R\},
\end{equation}
for some fixed $R>0$ and $s>2$. For $t>0$, let $S_t$ denote the data-to-solution map
\begin{equation}
L_s^\infty \rightarrow L_s^\infty, \qquad f_0 \mapsto f(t),
\end{equation}
where $f$ denotes the solution at time $t$ to the WKE \eqref{eq:WKE} with initial datum $f_0$. The work of Germain et al. \cite[Theorem 2.1]{GIT2020} shows that $f\in C([0,T];L_s^\infty)$ with
\begin{equation}
T\gtrsim_s R^{-2} \qquad \text{and} \qquad \sup_{0\leq t\leq T} \|f(t)\|_{L_s^\infty} \leq 2R.
\end{equation}
Moreover, the map $f_0\mapsto S_t f_0$ is continuous. Now there exists some $T\gtrsim_s R^{-2}$ such that
\begin{equation}
f^{(m)}(t) \coloneqq \int_{L_s^\infty}d\rho(f_0) (S_tf_0)^{\otimes m},
\end{equation}
defines a solution to the hierarchy \eqref{eq:WKE_hier} in $C([0,T];\Lr_{s,\ep}^\infty)$, for any $\ep < 1/(2R)$, and with initial datum
\begin{equation}
f_0^{(m)}\coloneqq \int_{L_s^\infty}d\rho(f_0) f_0^{\otimes m}.
\end{equation}
By \cref{thm:main} $(f^{(m)})_{m=1}^\infty$ is the unique solution. Thus, we have shown that
\begin{equation}
f^{(m)}(t) = \int_{L_s^\infty}d(S_t\# \rho)(f_0) f_0^{\otimes m},
\end{equation}
where $S_t\#\rho$ denotes the pushforward of $\rho$ under the map $S_t$.

\begin{remark}
Nowhere in the proof of \cref{thm:main} did we invoke an admissibility condition of the form \eqref{eq:adm}. Our result holds for general initial data $F_0$ belonging to the space $\Lr_{s,\ep}^\infty$, for $s>2$. These non-admissible solutions, though, do not seem to be physically meaningful for the reasons discussed above.
\end{remark}

\bibliographystyle{siam}
\bibliography{WT}

\end{document}